\documentclass[letterpaper]{article} 

\usepackage{aaai24}  
\usepackage{times}  
\usepackage{helvet}  
\usepackage{courier}  
\usepackage[hyphens]{url}  
\usepackage{graphicx} 
\usepackage{tikz}
\usepackage{todonotes}
\usepackage{natbib}  
\usepackage{caption} 
\usepackage{longtable}
\usepackage{caption}
\usepackage{subcaption}
\usepackage{thmtools}
\usepackage{thm-restate}
\usepackage{newfloat}
\usepackage{listings}
\DeclareCaptionStyle{ruled}{labelfont=normalfont,labelsep=colon,strut=off} 
\usepackage[ruled,linesnumbered,vlined,noend]{algorithm2e}
\usepackage{cite}
\usepackage{amsmath,amssymb,amsfonts}
\usepackage{textcomp}
\usepackage{xcolor}
\usepackage[switch]{lineno}
\usepackage{comment}
\usepackage{amssymb}
\usepackage{enumitem}
\usepackage[raggedrightboxes]{ragged2e}
\usepackage{amsfonts,amsthm,mathrsfs,ifpdf}
\usepackage{indentfirst,relsize}
\usepackage{enumerate}
\usepackage{latexsym}
\usepackage{stackrel}
\usepackage[all]{xy}
\usepackage[usenames,dvipsnames]{pstricks}
\usepackage{pst-grad} 
\usepackage{pst-plot} 
\usepackage{xspace}
\usepackage{bm, xurl}
\usepackage[super]{nth}
\usepackage[T1]{fontenc}
\usepackage{cleveref}

\title{Testing Self-Reducible Samplers}

\author{
Rishiraj Bhattacharyya\textsuperscript{\rm 1}\equalcontrib,
Sourav Chakraborty
\textsuperscript{\rm 2}\equalcontrib,
Yash Pote
\textsuperscript{\rm 3,4}\equalcontrib,
Uddalok Sarkar
\textsuperscript{\rm 2}\equalcontrib,
Sayantan Sen\textsuperscript{\rm 3}\equalcontrib
 }
\affiliations{

     \textsuperscript{\rm 1}University of Birmingham\\
     \textsuperscript{\rm 2} Indian Statistical Institute Kolkata\\
     \textsuperscript{\rm 3} National University of Singapore\\
      \textsuperscript{\rm 4} CNRS@CREATE\\
 }

\newcommand{\size}[1]{\left| #1 \right|}

\newcommand{\E}{\mathop{\mathbb{E}}}
\newcommand{\remove}[1]{}

\newcommand{\cI}{\mathcal{I}}
\newcommand{\cM}{\mathcal{M}}
\newcommand{\cS}{\mathcal{S}}

\newcommand{\cL}{\mathcal{L}}
\newcommand{\cA}{\mathcal{A}}

\newcommand{\cD}{\mathcal{D}}

\newcommand{\cP}{\mathcal{P}}

\newcommand{\cG}{\mathcal{G}}

\newcommand{\cW}{\mathcal{W}}
\newcommand{\Oh}{\mathcal{O}}

\newcommand{\eps}{\varepsilon}

\theoremstyle{plain}
\newtheorem{theo}{Theorem}
\newtheorem{lem}[theo]{Lemma}

\newtheorem{cl}[theo]{Claim}
\theoremstyle{definition}
\newtheorem{defi}[theo]{Definition}

\newcommand{\lcms}{\texttt{LxtCMSGen}}
\newcommand{\lqck}{\texttt{LxtQuickSampler}}
\newcommand{\lsts}{\texttt{LxtSTS}}
\newcommand*{\permcomb}[4][0mu]{{{}^{#3}\mkern#1#2_{#4}}}
\newcommand*{\perm}[1][-3mu]{\permcomb[#1]{P}}
\newcommand*{\comb}[1][-1mu]{\permcomb[#1]{C}}

\newcommand{\algoname}[1]{\ensuremath{\mathsf{#1}}}

\newcommand{\accept}{ACCEPT }
\newcommand{\reject}{REJECT }

\newcommand{{\dinf}}{{d_\infty}}

\newcommand{\subcond}{$\mathsf{SubCond}$ }

\newcommand{\gbas}{\algoname{GBAS} }

\newcommand{\dtv}{\mathsf{d_{TV}} }
\newcommand{\dtvp}{\mathsf{d^\psi_{TV}} }

\newcommand{\estimate}{\algoname{Est} }
\newcommand{\bk}{$\mathsf{Barbarik}$\xspace}
\newcommand{\wbk}{$\mathsf{Barbarik2}$\xspace}
\newcommand{\wbkk}{$\mathsf{Barbarik3}$\xspace}
\newcommand{\teq}{$\mathsf{Teq}$\xspace}
\newcommand{\whst}{Self-reducible-sampler-tester}

\newcommand{\iws}{\cI_{\cW}\xspace}
\newcommand{\igs}{\cI_{\cG}\xspace}

\newcommand{\wfl}{$\mathsf{Flash}$\xspace}

\newcommand{\gst}{\ensuremath{\mathsf{CubeProbeEst}}\xspace}
\newcommand{\gstt}{\ensuremath{\mathsf{CubeProbeTester}}\xspace}
\newcommand{\bern}{\cP}

\begin{document}

\maketitle

\begin{abstract}
Samplers are the backbone of the implementations of any randomised algorithm. Unfortunately, obtaining an efficient algorithm to test the correctness of samplers is very hard to find. Recently, in a series of works, testers like $\mathsf{Barbarik}$, $\mathsf{Teq}$, $\mathsf{Flash}$ for testing of some particular kinds of samplers, like CNF-samplers and Horn-samplers, were obtained. But their techniques have a significant limitation because one can not expect to use their methods to test for other samplers, such as perfect matching samplers or samplers for sampling linear extensions in posets. 
In this paper, we present a new testing algorithm that works for such samplers and can estimate the distance of a new sampler from a known sampler (say, uniform sampler).  

Testing the identity of distributions is the heart of testing the correctness of samplers. This paper's main technical contribution is developing a new distance estimation algorithm for distributions over high-dimensional cubes using the recently proposed sub-cube conditioning sampling model. Given subcube conditioning access to an unknown distribution $P$, and a known distribution $Q$ defined over $\{0,1\}^n$, our algorithm $\mathsf{CubeProbeEst}$ estimates the variation distance between $P$ and $Q$ within additive error $\zeta$ using $\mathcal{O}\left({n^2}/{\zeta^4}\right)$ subcube conditional samples from $P$.  Following the testing-via-learning paradigm, we also get a tester which distinguishes between the cases when $P$ and $Q$ are $\varepsilon$-close or $\eta$-far in variation distance with probability at least $0.99$ using $\mathcal{O}({n^2}/{(\eta-\varepsilon)^4})$ subcube conditional samples.

The estimation algorithm in the sub-cube conditioning sampling model helps us to design the first tester for self-reducible samplers. The correctness of the testers is formally proved. On the other hand, we implement our algorithm to create $\mathsf{CubeProbeEst}$ and use it to test the quality of three samplers for sampling linear extensions in posets. 

\end{abstract}

\section{Introduction}

\noindent Sampling algorithms play a pivotal role in enhancing the efficiency and accuracy of data analysis and decision-making across diverse domains~\cite{chandra1992constraint,yuan2004simplifying,naveh2007constraint,mironov2006applications,soos2009extending,morawiecki2013sat, ashur2017automated}. With the exponential surge in data volume, these algorithms provide the means to derive meaningful insights from massive datasets without the burden of processing the complete information. Additionally, they aid in pinpointing and mitigating biases inherent in data, ensuring the attainment of more precise and equitable conclusions. From enabling statistical inferences to propelling advancements in machine learning, safeguarding privacy, and facilitating real-time decision-making, sampling algorithms stand as a cornerstone in extracting information from the vast data landscape of our modern world.

However, many advanced sampling algorithms are often prohibitively slow (hash-based techniques of \cite{chakraborty2013scalable, ermon2013embed, chakraborty2014distribution, meel2015constrained} and MCMC-based methods of~\cite{andrieu2003introduction, brooks2011handbook, jerrum1998mathematical}) or lack comprehensive verification (\cite{ermon2012uniform}, \cite{dutra2018efficient}, \cite{GSCM21}). Many popular methods like ``statistical tests'' rely on heuristics without guarantees of their efficacy. Utilizing unverified sampling algorithms can lead to significant pitfalls, including compromised conclusion accuracy, potential privacy, and security vulnerabilities. Moreover, the absence of verification hampers transparency and reproducibility, underscoring the critical need for rigorous validation through testing, comparison, and consideration of statistical properties. Consequently, a central challenge in this field revolves around designing tools to certify sampling quality and verify correctness, which necessitates overcoming the intricate task of validating probabilistic programs and ensuring their distributions adhere to desired properties.

A notable breakthrough in addressing this verification challenge was achieved by \cite{chakraborty2019testing}, who introduced the statistical testing framework known as ``\bk''. This method proved instrumental in testing the correctness of uniform CNF (Conjunctive Normal Form) samplers by drawing samples from conditional distributions. \bk demonstrated three key properties: accepting an almost correct sampler with high probability, rejecting a far-from-correct sampler with high probability, and rejecting a ``well-behaved'' but far-from-correct sampler with high probability. There have been a series of follow-up works~\cite{meel2020testing, PM21, PM22, pmlr-v206-banerjee23a}. However, in this framework, conditioning is achieved using a gadget that does not quite generalize to applications beyond CNF sampling. For instance, for linear-extension sampling~\cite{huber2014near}, where the goal is to sample a linear ordering agreeing with a given poset,  the test requires that the post-conditioning residual input be a supergraph of the original input, with the property that it has exactly two \emph{user-specified} linear-extensions. This requirement is hard to fulfill in general. On the other hand, a generic tester that would work for any sampler implementation without any additional constraints and simultaneously be sample efficient is too good to be true \cite{tit/Paninski08}. From a practical perspective, the question is: \emph{Can we design an algorithmic framework for testers that would work for most deployed samplers and still have practical sample complexity?}

We answer the question positively. We propose algorithms that offer a generic approach to estimating the distance between a known and an unknown sampler, assuming both follow the ubiquitous self-reducible sampling strategy. Our techniques follow a constrained sampling approach, extending its applicability to wide range of samplers without mandating such specific structural conditions. A key foundational contribution of this paper includes leveraging the subcube conditional sampling techniques \cite{bhattacharyya2018property} and devising a method to estimate the distance between samplers – a challenge often more intricate than simple correctness testing.

\noindent\textbf{Organization of our paper} We first present the preliminaries followed by a description of our 
results and their relevance. We then give a detailed description of our main algorithms \gst and \gstt. The detailed theoretical analysis is presented in the supplementary material. We only present a high-level technical overview. Finally, we present our experimental results and conclude.

\section{Preliminaries}\label{sec:prelim}
\noindent In this paper, we are dealing with discrete probability distributions whose sample space is an $n$-dimensional Boolean hypercube, $\{0,1\}^n$. For a distribution $\mathcal{D}$ over a universe $\Omega$, and for any $x \in \Omega$,  we denote by $\mathcal{D}(x)$ the probability mass of the point $x$ in $\mathcal{D}$. 
$[n]$ denotes the set $\{1, \ldots,n\}$.
For concise expressions and readability, we use the asymptotic complexity notion of $\widetilde{\Oh}$, where we hide polylogarithmic dependencies of the parameters.   

\paragraph{Samplers, Estimators, and Testers}
A sampler $\cI: \mathsf{Domain} \to \mathsf{Range}$ is a randomized algorithm which, given an input $x\in \mathsf{Domain}$, outputs an element in $\mathsf{Range}$. For a sampler $\cI$, $\mathcal{D}^{\cI, \psi}$ denotes the probability distribution of the output of $\cI$ when the input is $\psi \in \mathsf{Domain}$. In other words, $$\forall x \in \mathsf{Range}, \ \cD^{\cI, \psi}(x) = \Pr[\cI(\psi) = x], $$
where the probability is over the internal random coins of $\cI$. 

\noindent We define a sampler $\iws$ to be a \emph{known sampler} if, for any input $\psi \in \mathsf{Domain}$, we know its probability distribution $\cD^{\iws, \psi}$ explicitly. 
We note that the input $\psi$ depends on the application. For example, in the perfect-matching and linear-extension samplers,  $\psi$ is a graph, whereas, for the CNF sampler, $\psi$ is a CNF formula.

\begin{defi}[{\bf Total variation distance}]\label{defi:tvd}
Let $\iws$ and $\igs$ be two samplers. For an input $\psi \in \mathsf{Domain}$, the variation distance between $\igs$ and $\iws$ is defined as: 
$$\dtvp(\igs, \iws) = \max_{A \subseteq \mathsf{Range}} \{\cD^{\igs, \psi}(A) - \cD^{\iws, \psi}(A)\}.$$
\end{defi}

\begin{defi}[{\bf $(\zeta, \delta$)-approx $\dtv$ estimator}] 
A $(\zeta, \delta$)-approx $\dtv$ estimator is a randomized approximation algorithm that given two sampler $\igs$ and $\iws$, an input $\psi$, tolerance parameter $\zeta \in (0,1/3]$ and a confidence parameter $\delta\in (0,1)$, with probability $(1-\delta)$ returns an estimation $\widehat{\mathsf{dist}^{\psi}}$ of $\dtvp(\igs,\iws)$ such that:
$$\dtvp(\igs,\iws) - \zeta \leq \widehat{\mathsf{dist}^{\psi}} \leq \dtvp(\igs,\iws) + \zeta$$
\end{defi}

\begin{defi}[{\bf $\eps$-closeness and $\eta$-farness}]\label{defi:wclosefar}
Consider any sampler $\igs$. $\igs$ is said to be $\eps$-close to another sampler $\iws$ on input $\psi$, if $\dtvp\left(\igs,\iws\right) \leq \eps$ holds. On the other hand, $\igs$ is said to be $\eta$-far from $\iws$ with respect to some input $\psi$ if $\dtvp\left(\igs,\iws\right) \geq \eta$ holds. 
\end{defi}

\begin{defi}[{\bf $(\eps, \eta, \delta)$-identity tester)}]
An $(\eps, \eta, \delta)$-identity tester takes as input an unknown sampler $\igs$, a known sampler $\iws$, an input $\psi$ to the samplers, a tolerance parameter $ \eps \in (0, 1/3)$, an intolerance parameter $\eta \in (0,1]$ with $ \eta> \eps$, a confidence parameter $\delta \in (0,1)$, and with probability at least $(1 - \delta)$: (1) outputs \accept if $\cI_{\cG}$ is $\eps$-close to $\iws$ on input $\psi$, (2) outputs \reject if $\cI_{\cG}$ is $\eta$-far from $\iws$ on input $\psi$.
\end{defi}
\noindent For practical purposes, $\delta$ can be $0.99$ or any close-to-one constant.
From now onwards,  we shall consider the input domain and output range of a sampler to be a Boolean hypercube, that is, $\mathsf{Domain} = \{0,1\}^m$ and $\mathsf{Range} = \{0,1\}^n$ for some integers $m$ and $n$. Therefore the universe of probability distributions of samplers is $n$-dimensional binary strings.

\paragraph{Self-reducible sampler.}
A self-reducible sampler $\cI:\{0,1\}^m \to \{0,1\}^n$ generates a sample $x$ by first sampling a bit and then sampling the rest of the substring. Formally, we can define a self-reducible sampler as follows:
\begin{defi}[{\bf Self-reducible sampler}]
    A sampler $\cI:\{0,1\}^m \to \{0,1\}^n$ is said to be a self-reducible sampler if, for any input $\psi \in \{0,1\}^m$, there exists $\widehat{\psi} \in \{0,1\}^m$ for which the following is true: 
    \[\cD^{\cI, \psi}(x_1x_2...x_n)|_{x_1 = b_1, ..., x_i = b_i} = \cD^{\cI, \widehat{\psi}}(b_1...b_ix_{i+1}...x_n)\]
    where $b_i\in\{0,1\}$ for all $i$.
\end{defi}
\noindent The concept of self-reducibility has been influential in the field of sampling since the work of \cite{jerrum1986random}, which showed the computational complexity equivalence of approximate sampling and counting for problems in \#P. Intuitively, self-reducibility is the idea that one can construct the solution to a given problem from the solutions of subproblems of the same problem. Self-reducibility is a critical requirement for simulating subcube conditioning. Also, it does not hamper the model's generality too much. As observed in \cite{khuller1991planar, grosse2006computing,talvitie2020exact}, all except a few known problems are self-reducible. 

\paragraph{Subcube Conditioning over Boolean hypercubes}
 Let $P$ be a probability distribution over $\{0,1\}^n$. Sampling using subcube conditioning accepts $A_1, A_2, \ldots, A_n \subseteq \{0,1\}$, constructs $S = A_1 \times A_2 \times \ldots \times A_n$ as the condition set, and returns a vector $x = (x_1, x_2, \ldots,x_n)$, such that $x_i \in A_i$, with probability $P(x)/(\sum_{w \in S} P(w))$. If $P(S)$ = 0, we assume the sampling process would return an element from $S$ uniformly at random. A sampler that follows this technique is called a subcube conditioning sampler.

\subsubsection{Linear-Extension of a Poset}
We applied our prototype implementation on verifying linear-extension samplers of a poset. Let us first start with the definition of a poset.

\begin{defi}[{\bf Partially ordered set (Poset)}]
Let $S$ be a set on $k$ elements. A relation $\preceq$ (subset of $S \times S$) is said to be a partial order if $\preceq$ is (i) reflexive ($a \preceq a$ for every $a \in S$) (ii) anti-symmetric ($a \preceq b$ and $b \preceq a$ implies $a=b$ for every $a,b \in S$) and (iii) transitive ($a \preceq b$ and $b \preceq c$ implies $a \preceq c$ for every $a,b,c \in S$). We say $(S, \preceq)$ is a \emph{partially ordered set} or \emph{poset} in short.
If all pairs of $S$ are comparable, that is, for any $a, b \in S$, either $a \preceq b$ or $b \preceq a$ then $(S, \preceq)$ is called a \emph{linear ordered set}.
\end{defi}

\begin{defi}[{\bf Linear-extension of poset}]
A relation $\preceq_l \ \supseteq \ \preceq$ is called a linear-extension of $\preceq$, if $(S,\preceq_l)$ is linearly ordered. Given a poset $\cP=(S,\preceq)$, we denote the set of all possible linear-extensions by $\cL(\cP)$.
\end{defi}

\begin{defi}[{\bf Linear-extension sampler}]
Given a poset $\cP = (S, \preceq)$, a linear-extension sampler $\cI_{Lext}$ samples a possible linear-extension $\preceq_l$ of $\cP$ from the set of all possible linear-extensions $\cL(\cP)$. 
\end{defi}

\paragraph{Linear-extension to Boolean Hypercube}  Let us define a \emph{base} linear ordering on $S$ as $\preceq_l'$. We order the elements of $S$ as $S_1 \preceq_l' S_2 \preceq_l' ... \preceq_l' S_k$ based on $\preceq_l'$, where $k=|S|$. For a poset $\cP = (S, \preceq)$, we construct a $k \times k$ matrix $\cM_\cP$ such that for all $i$, $\cM_\cP(i,i) := 1$ and for all $i \neq j$, if $S_i \preceq S_j$ then $\cM_\cP(i,j) := 1$ and $\cM_\cP(i,j) := 0$ when $S_j \preceq S_i$, if $(S_i, S_j) \notin \preceq$, that is if $S_i, S_j$ are not comparable in $\preceq$, then $\cM_\cP(i,j) := *$. The matrix $\cM_\cP$ is a unique representation of the poset $\cP = (S, \preceq)$. $\cM_\cP$ is anti-symmetric, i.e., the upper triangle of $\cM_\cP$ is exactly the opposite of the lower triangle (apart from the $*$ and the diagonal entries). So only the upper triangle of $\cM_\cP$ without the diagonal entries can represent $\cP$. Now unrolling of the upper triangle of $\cM_\cP$ (without the diagonal) creates a $\{0,1,*\}^{\comb{k}{2}}$ string $x_{\cM_\cP}$. Suppose for a $\cP$ there are $n$ $*$'s in the unrolling. Then we can say sampling a linear-extension of $\cP$ is equivalent to sampling from a $\{0,1\}^n$ subcube of the Boolean hypercube $\{0,1\}^{\comb{k}{2}}$, where $\cP$ induces subcube conditioning by fixing the bits of non-$*$ dimensions. Adding one more new pair, say $(S_{i'}, S_{j'})$, to $\cP$ results in fixing one more bit of $x_{\cM_\cP}$ and vice versa. We introduce a mapping $\mathsf{SubCond}$ that can incorporate a new pair into poset $\cP$ and subsequently fixes the corresponding bit in bit string $x_{\cM_\cP}$. Thus $\mathsf{SubCond}$ provides a method to achieve subcube conditioning on a poset.

\begin{figure}
    \centering
    \scalebox{0.5}{%

\tikzset{every picture/.style={line width=0.75pt}} %

\begin{tikzpicture}[x=0.6pt,y=0.6pt,yscale=-1,xscale=1]
\draw   (210.95,200.87) .. controls (205.42,200.85) and (200.96,196.36) .. (200.97,190.84) .. controls (200.98,185.32) and (205.47,180.85) .. (211,180.87) .. controls (216.52,180.88) and (220.98,185.37) .. (220.97,190.89) .. controls (220.96,196.41) and (216.47,200.88) .. (210.95,200.87) -- cycle ;
\draw   (170.95,135.5) .. controls (165.43,135.49) and (160.96,131) .. (160.98,125.48) .. controls (160.99,119.95) and (165.48,115.49) .. (171,115.5) .. controls (176.52,115.51) and (180.99,120) .. (180.98,125.52) .. controls (180.96,131.05) and (176.47,135.51) .. (170.95,135.5) -- cycle ;
\draw   (255.24,137.58) .. controls (249.72,137.57) and (245.25,133.08) .. (245.27,127.56) .. controls (245.28,122.04) and (249.77,117.57) .. (255.29,117.58) .. controls (260.81,117.6) and (265.28,122.08) .. (265.26,127.61) .. controls (265.25,133.13) and (260.76,137.6) .. (255.24,137.58) -- cycle ;
\draw   (211.51,71.25) .. controls (205.99,71.24) and (201.52,66.75) .. (201.53,61.23) .. controls (201.55,55.7) and (206.03,51.24) .. (211.56,51.25) .. controls (217.08,51.26) and (221.55,55.75) .. (221.53,61.27) .. controls (221.52,66.8) and (217.03,71.26) .. (211.51,71.25) -- cycle ;
\draw    (216.5,182.5) -- (249.9,139.45) ;
\draw [shift={(251.74,137.08)}, rotate = 127.81] [fill={rgb, 255:red, 0; green, 0; blue, 0 }  ][line width=0.08]  [draw opacity=0] (8.93,-4.29) -- (0,0) -- (8.93,4.29) -- cycle    ;
\draw    (204,182.5) -- (175.6,138.36) ;
\draw [shift={(173.98,135.84)}, rotate = 57.24] [fill={rgb, 255:red, 0; green, 0; blue, 0 }  ][line width=0.08]  [draw opacity=0] (8.93,-4.29) -- (0,0) -- (8.93,4.29) -- cycle    ;
\draw    (171,115.5) -- (201.86,69.76) ;
\draw [shift={(203.53,67.27)}, rotate = 124] [fill={rgb, 255:red, 0; green, 0; blue, 0 }  ][line width=0.08]  [draw opacity=0] (8.93,-4.29) -- (0,0) -- (8.93,4.29) -- cycle    ;

\draw [color={rgb, 255:red, 139; green, 87; blue, 42 }  ,draw opacity=1 ][line width=3]    (180.98,318.52) -- (237.5,318.5) ;
\draw [shift={(243.5,318.5)}, rotate = 179.98] [fill={rgb, 255:red, 139; green, 87; blue, 42 }  ,fill opacity=1 ][line width=0.08]  [draw opacity=0] (16.97,-8.15) -- (0,0) -- (16.97,8.15) -- cycle    ;
\draw   (209.95,393.87) .. controls (204.42,393.85) and (199.96,389.36) .. (199.97,383.84) .. controls (199.98,378.32) and (204.47,373.85) .. (210,373.87) .. controls (215.52,373.88) and (219.98,378.37) .. (219.97,383.89) .. controls (219.96,389.41) and (215.47,393.88) .. (209.95,393.87) -- cycle ;
\draw   (169.95,328.5) .. controls (164.43,328.49) and (159.96,324) .. (159.98,318.48) .. controls (159.99,312.95) and (164.48,308.49) .. (170,308.5) .. controls (175.52,308.51) and (179.99,313) .. (179.98,318.52) .. controls (179.96,324.05) and (175.47,328.51) .. (169.95,328.5) -- cycle ;
\draw   (254.24,330.58) .. controls (248.72,330.57) and (244.25,326.08) .. (244.27,320.56) .. controls (244.28,315.04) and (248.77,310.57) .. (254.29,310.58) .. controls (259.81,310.6) and (264.28,315.08) .. (264.26,320.61) .. controls (264.25,326.13) and (259.76,330.6) .. (254.24,330.58) -- cycle ;
\draw   (210.51,264.25) .. controls (204.99,264.24) and (200.52,259.75) .. (200.53,254.23) .. controls (200.55,248.7) and (205.03,244.24) .. (210.56,244.25) .. controls (216.08,244.26) and (220.55,248.75) .. (220.53,254.27) .. controls (220.52,259.8) and (216.03,264.26) .. (210.51,264.25) -- cycle ;
\draw    (215.5,375.5) -- (248.9,332.45) ;
\draw [shift={(250.74,330.08)}, rotate = 127.81] [fill={rgb, 255:red, 0; green, 0; blue, 0 }  ][line width=0.08]  [draw opacity=0] (8.93,-4.29) -- (0,0) -- (8.93,4.29) -- cycle    ;
\draw    (203,375.5) -- (174.6,331.36) ;
\draw [shift={(172.98,328.84)}, rotate = 57.24] [fill={rgb, 255:red, 0; green, 0; blue, 0 }  ][line width=0.08]  [draw opacity=0] (8.93,-4.29) -- (0,0) -- (8.93,4.29) -- cycle    ;
\draw    (170,308.5) -- (200.86,262.76) ;
\draw [shift={(202.53,260.27)}, rotate = 124] [fill={rgb, 255:red, 0; green, 0; blue, 0 }  ][line width=0.08]  [draw opacity=0] (8.93,-4.29) -- (0,0) -- (8.93,4.29) -- cycle    ;

\draw [color={rgb, 255:red, 139; green, 87; blue, 42 }  ,draw opacity=1 ][line width=3]    (203,375.5) -- (176.22,333.89) ;
\draw [shift={(172.98,328.84)}, rotate = 57.24] [fill={rgb, 255:red, 139; green, 87; blue, 42 }  ,fill opacity=1 ][line width=0.08]  [draw opacity=0] (16.97,-8.15) -- (0,0) -- (16.97,8.15) -- cycle    ;
\draw [color={rgb, 255:red, 139; green, 87; blue, 42 }  ,draw opacity=1 ][line width=3]    (250,310.5) -- (220.84,266.98) ;
\draw [shift={(217.5,262)}, rotate = 56.17] [fill={rgb, 255:red, 139; green, 87; blue, 42 }  ,fill opacity=1 ][line width=0.08]  [draw opacity=0] (16.97,-8.15) -- (0,0) -- (16.97,8.15) -- cycle    ;

\draw [color={rgb, 255:red, 139; green, 87; blue, 42 }  ,draw opacity=1 ][line width=3]    (462.27,319.06) -- (405.5,319.46) ;
\draw [shift={(399.5,319.5)}, rotate = 359.6] [fill={rgb, 255:red, 139; green, 87; blue, 42 }  ,fill opacity=1 ][line width=0.08]  [draw opacity=0] (16.97,-8.15) -- (0,0) -- (16.97,8.15) -- cycle    ;
\draw   (428.95,393.37) .. controls (423.42,393.35) and (418.96,388.86) .. (418.97,383.34) .. controls (418.98,377.82) and (423.47,373.35) .. (429,373.37) .. controls (434.52,373.38) and (438.98,377.87) .. (438.97,383.39) .. controls (438.96,388.91) and (434.47,393.38) .. (428.95,393.37) -- cycle ;
\draw   (388.95,328) .. controls (383.43,327.99) and (378.96,323.5) .. (378.98,317.98) .. controls (378.99,312.45) and (383.48,307.99) .. (389,308) .. controls (394.52,308.01) and (398.99,312.5) .. (398.98,318.02) .. controls (398.96,323.55) and (394.47,328.01) .. (388.95,328) -- cycle ;
\draw   (473.24,330.08) .. controls (467.72,330.07) and (463.25,325.58) .. (463.27,320.06) .. controls (463.28,314.54) and (467.77,310.07) .. (473.29,310.08) .. controls (478.81,310.1) and (483.28,314.58) .. (483.26,320.11) .. controls (483.25,325.63) and (478.76,330.1) .. (473.24,330.08) -- cycle ;
\draw   (429.51,263.75) .. controls (423.99,263.74) and (419.52,259.25) .. (419.53,253.73) .. controls (419.55,248.2) and (424.03,243.74) .. (429.56,243.75) .. controls (435.08,243.76) and (439.55,248.25) .. (439.53,253.77) .. controls (439.52,259.3) and (435.03,263.76) .. (429.51,263.75) -- cycle ;
\draw    (434.5,375) -- (467.9,331.95) ;
\draw [shift={(469.74,329.58)}, rotate = 127.81] [fill={rgb, 255:red, 0; green, 0; blue, 0 }  ][line width=0.08]  [draw opacity=0] (8.93,-4.29) -- (0,0) -- (8.93,4.29) -- cycle    ;
\draw    (422,375) -- (393.6,330.86) ;
\draw [shift={(391.98,328.34)}, rotate = 57.24] [fill={rgb, 255:red, 0; green, 0; blue, 0 }  ][line width=0.08]  [draw opacity=0] (8.93,-4.29) -- (0,0) -- (8.93,4.29) -- cycle    ;
\draw    (389,308) -- (419.86,262.26) ;
\draw [shift={(421.53,259.77)}, rotate = 124] [fill={rgb, 255:red, 0; green, 0; blue, 0 }  ][line width=0.08]  [draw opacity=0] (8.93,-4.29) -- (0,0) -- (8.93,4.29) -- cycle    ;

\draw [color={rgb, 255:red, 139; green, 87; blue, 42 }  ,draw opacity=1 ][line width=3]    (434.5,375) -- (466.06,334.32) ;
\draw [shift={(469.74,329.58)}, rotate = 127.81] [fill={rgb, 255:red, 139; green, 87; blue, 42 }  ,fill opacity=1 ][line width=0.08]  [draw opacity=0] (16.97,-8.15) -- (0,0) -- (16.97,8.15) -- cycle    ;
\draw [color={rgb, 255:red, 139; green, 87; blue, 42 }  ,draw opacity=1 ][line width=3]    (389,308) -- (418.18,264.75) ;
\draw [shift={(421.53,259.77)}, rotate = 124] [fill={rgb, 255:red, 139; green, 87; blue, 42 }  ,fill opacity=1 ][line width=0.08]  [draw opacity=0] (16.97,-8.15) -- (0,0) -- (16.97,8.15) -- cycle    ;

\draw (207.4,52.6) node [anchor=north west][inner sep=0.75pt]  [font=\small]  {$4$};
\draw (166.5,118.6) node [anchor=north west][inner sep=0.75pt]  [font=\small]  {$2$};
\draw (251.1,118.5) node [anchor=north west][inner sep=0.75pt]  [font=\small]  {$3$};
\draw (205.9,182.9) node [anchor=north west][inner sep=0.75pt]  [font=\small]  {$1$};
\draw (204.9,375.9) node [anchor=north west][inner sep=0.75pt]  [font=\small]  {$1$};
\draw (250.1,311.5) node [anchor=north west][inner sep=0.75pt]  [font=\small]  {$3$};
\draw (165.5,311.6) node [anchor=north west][inner sep=0.75pt]  [font=\small]  {$2$};
\draw (206.4,245.6) node [anchor=north west][inner sep=0.75pt]  [font=\small]  {$4$};
\draw (423.9,375.4) node [anchor=north west][inner sep=0.75pt]  [font=\small]  {$1$};
\draw (469.1,311) node [anchor=north west][inner sep=0.75pt]  [font=\small]  {$3$};
\draw (384.5,311.1) node [anchor=north west][inner sep=0.75pt]  [font=\small]  {$2$};
\draw (425.4,245.1) node [anchor=north west][inner sep=0.75pt]  [font=\small]  {$4$};
\draw (370,78) node [anchor=north west][inner sep=0.75pt]   [align=left] {\LARGE $\begin{bmatrix}
1 & \bm{\underline 1} & \bm{\underline 1} & \bm{\underline 1}\\
0 & 1 & \bm{\underline *} & \bm{\underline 1}\\
0 & * & 1 & \bm{\underline *}\\
0 & 0 & * & 1\\
\end{bmatrix}$};

\end{tikzpicture}
}
    \caption{\small\justifying The top-left graph represents the cover graph of a poset $\cP = (S, \preceq)$ over $S = \{1,2,3,4\}$, and poset relation $\preceq = \{(1,2), (1,3), (2,4), (1,4)\}$. The bottom row shows two possible linear-extensions $1 \preceq 2 \preceq 3 \preceq 4$ and $1 \preceq 3 \preceq 2 \preceq 4$, and the corresponding cover graphs, drawn in red arrows. The matrix on the top-right corresponds to $\cM_\cP$. Unrolling of the upper triangle (bold-underline) of $\cM_\cP$ gives $x_{\cM_\cP} = 111*1*$. Fixing the 4th bit of $x_{\cM_\cP}$ to $0$ is equivalent to including the relation $(3,2)$ into $\preceq$. Here $\mathsf{SubCond}(\preceq, 4) = \preceq \cup \{(3,2)\}$.}
    \label{fig:lext}
\end{figure}
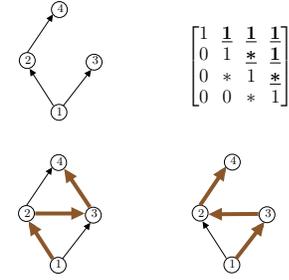

\paragraph{Basic Probability Facts}

We will use the following probability notations in our algorithm. A random variable $X$ is said to follow the exponential distribution with parameter $\lambda$ if $\Pr(X=x)=\lambda e^{-\lambda x}$ if $x \geq 0$ and $0$ otherwise. This is represented as $X \sim \mathsf{Exp}(\lambda)$. A random variable $X$ is said to be \emph{sub-Gaussian} (SubG in short) with parameter $\alpha^2$ if and only if its tails are dominated by a Gaussian of parameter $\alpha^2$. We include formal definitions and related concentration bounds in the supplementary material.

\section{Our Results}
\label{sec:ourresult}

\noindent The main technical contribution of this work is the algorithm \gst that can estimate the variation distance between a known and an unknown self-reducible sampler. The following informal theorem captures the details.

\begin{restatable}{theo}{cubeprobe}
\label{thm:main}
For an error parameter $\zeta \in (0,1)$, and a constant $\delta<1/3$, \gst is  $(\zeta, \delta$)-approx $\dtv$ estimator between a known and unknown self-reducible samplers $\iws$ and $\igs$ respectively with sample complexity of $\widetilde{\mathcal{O}}\left(n^2/\zeta^4\right)$.
\end{restatable}

\noindent Our framework seamlessly extends to yield an $(\eps, \eta, \delta)$-tester \gstt through the ``testing-via-learning'' paradigm \cite{diakonikolas2007testing,DBLP:conf/icalp/GopalanOSSW09,servedio2010testing}. To test whether the sampler's output distribution is $\eps$-close or $\eta$-far from the target output distribution, the resultant tester requires  $\widetilde{\mathcal{O}}\left(n^2/\left(\eta-\eps\right)^4\right)$ samples.

\noindent To demonstrate the usefulness of \gst, we developed a prototype implementation with experimental evaluations in gauging the correctness of linear-extension samplers while emulating uniform samplers. Counting the size of the set of linear extensions and sampling from them has been widely studied in a series of works by~\cite{huber2014near, talvitie2018scalable, talvitie2018counting}. The problem found extensive applications in artificial intelligence, particularly in learning graphical models \cite{wallace1996causal}, in sorting \cite{peczarski2004new}, sequence analysis \cite{mannila2000global}, convex rank tests \cite{morton2009convex}, preference reasoning \cite{lukasiewicz2014probabilistic}, partial order plans \cite{muise2016optimal} etc. Our implementation extends to a closeness tester that accepts ``close to uniform''  samplers and rejects ``far from uniform'' samplers. Moreover, while rejecting, our implementation can produce a certificate of non-uniformity. \gst and \gstt are the first estimator and tester for general self-reducible samplers.

\paragraph{Novelty in Our Contributions} In relation to the previous works, we emphasize our two crucial novel contributions.
\begin{itemize}
    \item Our algorithm is grounded in a notably refined form of ``grey-box'' sampling methodology, setting it apart from prior research endeavors \cite{chakraborty2019testing,meel2020testing,pmlr-v206-banerjee23a}. While prior approaches required arbitrary conditioning, our algorithm builds on the significantly weaker subcube conditional sampling paradigm~\cite{bhattacharyya2018property}. Subcube conditioning is a natural fit for ubiquitous self-reducible sampling, and thus our algorithm accommodates a considerably broader spectrum of sampling scenarios.

    \item All previous works produced testers crafted to produce a ``yes'' or ``no'' answer to ascertain correctness of samplers. In essence, these testers strive to endorse samplers that exhibit ``good'' behavior while identifying and rejecting those that deviate significantly from this standard. However, inherent technical ambiguity exists in setting the thresholds of the distances ($\eta$ and $\eps$) that would label a sampler as good or bad. In contrast, the \gst framework produces the estimated statistical distance that allows a practitioner to make informed and precise choices while selecting a sampler implementation.  In this context \gst is the first of its kind.
\end{itemize}

\paragraph{Our Contribution in the Context of Distribution Testing with Subcube Conditional Samples.}

The crucial component in designing our self-reducible-sampler-tester \gst is a novel algorithm for estimating the variation distance in the subcube conditioning model in distribution testing. Given sampling access to an unknown distribution $P$ and a known distribution $Q$ over $\{0,1\}^n$, the distance estimation problem asks to estimate the variation distance between $P$ and $Q$. The corresponding testing problem is the \emph{tolerant identity testing} of $P$ and $Q$. Distance estimation and tolerant testing with subcube conditional samples have been open since the introduction of the framework five years ago. The following theorem formalizes our result in the context of distance estimation/tolerant testing using subcube conditional samples. 
\begin{theo}
Let $P$ be an unknown distribution and $Q$ be a known distribution defined over $\{0,1\}^n$. Given subcube conditioning access to $P$, an approximation parameter $\gamma \in (0,1)$ and a confidence parameter $\delta \in (0,1)$, there exists an algorithm that takes $\widetilde{\Oh}(n^2)$ subcube-conditional samples from $P$ on expectation and outputs an estimate of $\dtv(P,Q)$ with an additive error $\zeta$ with probability at least $1-\delta$.
\end{theo}

\noindent This is the first algorithm that solves the variation distance estimation problem in $\widetilde{\Oh}(n^2)$ subcube conditioning samples.

\subsection{Related Works}
\noindent The state-of-the-art approach for efficiently testing CNF samplers was initiated by Meel and Chakraborty~\cite{chakraborty2019testing}. They employed the concept of hypothesis testing with conditional samples~\cite{chakraborty2016power,canonne2015testing} and showed that such samples could be ``simulated'' in the case of CNF samplers. The approach produced mathematical guarantees on the correctness of their tester. Their idea was extended to design a series of testers for various types of CNF samplers (\bk~\cite{chakraborty2019testing} for uniform CNF samplers, \wbk~\cite{meel2020testing} for weighted CNF samplers, \teq~\cite{PM21} for testing probabilistic circuits, \wfl~\cite{pmlr-v206-banerjee23a} for Horn samplers, \wbkk~\cite{PM22} for constrained samplers).

The theoretical foundation of our work follows the \emph{subcube conditioning} model of property testing of probability distributions. This model was introduced by  \cite{bhattacharyya2018property} as a special case of the conditional sampling model \cite{chakraborty2016power,canonne2015testing} targeted towards high-dimensional distributions. Almost all the known results in the subcube conditioning framework deal with problems in the non-tolerant regime: testing uniformity, identity, and equivalence of distributions. \cite{canonne2021random} presented optimal algorithm for (non-tolerant) uniformity testing in this model. \cite{chen2021learning} studied the problem of learning and testing junta distributions. Recently \cite{mahajan2023learning} studied the problem of learning Hidden Markov models. \cite{blanca2023complexity} studied identity testing in related coordinate conditional sampling model. \cite{fotakis2020efficient} studied parameter estimation problem for truncated Boolean product distributions. Recently \cite{chen2023uniformity} studied the problem of uniformity testing in hypergrids. Very recently, in a concurrent work, the authors in \cite{kumar2023tolerant} studied the problem of tolerant equivalence testing where both the samplers are unknown and designed an algorithm that takes $\widetilde{\Oh}(n^3)$ samples.

\section{Estimator of Self-reducible Samplers}
\label{sec:test-graph-sampl}

\noindent Our estimator utilizes the subcube conditional sampling technique. The main program \gst works with two subroutines: $\estimate$ and $\gbas$. The algorithm $\gbas$ is adopted from the \emph{Gamma Bernoulli Approximation Scheme} \cite{huber2014near}. Since its intricacies are crucial for our algorithm, we include the algorithm here for completeness.

\begin{algorithm}[ht]
\caption{\gst($\igs, \iws, \psi, \zeta, \delta$)}\label{balg:indentl1}

$\alpha =\frac{2}{\zeta^2} \log \frac{4}{\delta}$ \label{line:toltester:line1} \

$\gamma = \frac{\zeta}{1.11(2+ \zeta)}$ \label{line:toltester:line2} \

$\delta' = \delta/2\alpha$ \label{line:toltester:line3} \

$S= \emptyset$ \label{line:toltester:line6}

$S \gets \alpha \ \mbox{iid samples from} \ \igs(\psi)$ \label{line:toltester:line8} \

$val=0$ \label{line:toltester:line9}\

\For{$x \in S$}{\label{line:toltester:forloop:line10}

$val = val+ \max\left(0, 1 - \frac{\cD^{\iws,\psi}(x)}{\estimate(\igs, \psi, n, x, \gamma, \delta')} \right)$ \label{line:toltester:line11}\

}

\Return $\frac{val}{\alpha}$\label{line:toltester:line12}
                           
\end{algorithm}

\begin{algorithm}[ht]
\caption{\estimate($\igs, \psi, n, x, \gamma, \delta'$)}\label{balg:subestimate}

{ $k \gets \lceil \frac{3n}{\gamma^2} \cdot \log\left(\frac{2n}{\delta'}\right) \rceil$ \
} 
\

\For{$i=1 \ to \ n$}{\label{line:estimate:forloop:line1}

$\widehat{\psi} \gets \text{\subcond}(\psi, x_{1} ... x_{i-1} )$

$\widehat{P}_{i} \gets \gbas (\igs, \widehat{\psi}, i, k, x_i)$ \label{line:estimate:line2} \

}

$\widehat{\cD^{\igs, \psi}_x} = \Pi_{i=1}^n \widehat{P_{i}}$ \label{line:estimate:line3} \ 

\Return $\widehat{\cD^{\igs, \psi}_x}$ \label{line:estimate:line4} \

\end{algorithm}

\begin{algorithm}[ht]
    \caption{\gbas($\igs$, $\widehat{\psi}$, $i$, $k$, $\mathtt{HEAD}$)}
    \label{alg:gbas}
    $s \gets 0, \, r \gets 0$\;\label{line:gbas:line1}
    \While{$s < k$}{\label{line:gbas:whileloop:line2}
    $w \sim \igs(\widehat{\psi})$\;\label{line:gbas:whileloop:line3}
   \If{$\mathtt{HEAD} = w_i$}{
   \label{line:gbas:whileloop:line4}
        $s \gets s + 1$ \;
        \label{line:gbas:whileloop:line5}
    }
    $a \sim \mathsf{Exp}(1)$, $r \gets r + a$ \;
    \label{line:gbas:whileloop:line6}
    }
    $\widehat{p} \gets (k - 1)/r$ \;\label{line:gbas:line7}
    \Return $\widehat{p}$\; \label{line:gbas:line8}
\end{algorithm}

\vspace{.5em}

\noindent \bm{$\mathsf{CubeProbeEst}$}: In this algorithm, given a known self-reducible sampler $\iws$, subcube conditioning access to an unknown self-reducible sampler $\igs$, along with an input $\psi$, an approximation parameter $\zeta$ and a confidence parameter $\delta$, it estimates the variation distance between $\igs$ and $\iws$ with additive error $\zeta$. \gst uses the algorithm \estimate as a subroutine. It starts by setting several parameters $\alpha,\gamma, \delta'$ in \Cref{line:toltester:line1}-\Cref{line:toltester:line3}. In \Cref{line:toltester:line6}, it initializes an empty multi-set $S$, and then takes $\alpha$ samples from $\igs(\psi)$ in $S$ in \Cref{line:toltester:line8}. Now it defines a counter $val$ in \Cref{line:toltester:line9}, initialized to $0$. Now in the for loop starting from \Cref{line:toltester:forloop:line10}, for every sample $x \in S$ obtained before, \gst calls the subroutine \estimate in \Cref{line:toltester:line11} to estimate the probability mass of $\cD^{\iws, \psi}$ at $x$. Finally, in \Cref{line:toltester:line12}, we output $val/\alpha$ as the estimated variation distance and terminate the algorithm.

\vspace{.5em}

\noindent \bm{$\mathsf{Est}$}: Given subcube conditioning access to the unknown self-reducible sampler $\igs$, an input $\psi$, the dimension $n$, an $n$-bit string $x$, parameters $\gamma$ and $\delta'$ and an integer $t$, the subroutine \estimate returns an estimate of the probability of $\cD^{\igs, \psi}$ at $x$ by employing the subroutine \gbas. In the for loop starting from \Cref{line:estimate:forloop:line1}, it first calls \subcond with $\psi$ and $x_1, \ldots, x_{i-1}$ which outputs $\widehat{\psi}$. Now in \Cref{line:estimate:line2} it calls \gbas  with $\iws, \widehat{\psi}, i , k$  along with the $i$-th bit of $x$, i.e , $x_i$ with the integer $k$ (to be fixed such that $\delta'/n= 2\exp (-k \gamma^2/3)$) to estimate $\widehat{P_i}$, the empirical weight of $\cD^{\igs, \widehat{\psi}}$. Now in \Cref{line:estimate:line3}, \estimate computes the empirical weight of $\cD^{\igs, \psi}(x)$ by taking a product of all marginal distributions $\widehat{P}_1, \ldots, \widehat{P}_n$ obtained from the above for loop. Finally in \Cref{line:estimate:line4}, \estimate returns $\widehat{\cD_x^{\igs, \psi}}$, the estimated weight of the distribution $\cD^{\igs, \psi}$ on $x$.

\vspace{.5em} 

\noindent \bm{$\mathsf{GBAS}$}: In this algorithm, given access to an unknown self-reducible sampler $\igs$, input $\widehat{\psi}$, integers $i$ and $k$, and a bit $\mathtt{HEAD}$, \gbas outputs an estimate $\widehat{p}$ of $p$. \gbas starts by declaring two variables $s$ and $r$, initialized to $0$ in \Cref{line:gbas:line1}. Then in the for loop starting in \Cref{line:gbas:whileloop:line2}, as long as $s <k$, it first takes a sample $w$ from the sampler $\igs$ on input $\widehat{\psi}$ in \Cref{line:gbas:whileloop:line3}. Then in \Cref{line:gbas:whileloop:line4}, it checks if the value of $\mathtt{HEAD}$ is $w_i$ where $w_i$ is the $i$-th bit of the $n$-bit sample $w$. If the value of $\mathtt{HEAD}$ equals $w_i$, then in \Cref{line:gbas:whileloop:line5}, it increments the value of $s$ by 1. Then in \Cref{line:gbas:whileloop:line6}, \gbas samples  $a$ following $\mathsf{Exp(1)}$, the exponential distribution with parameter $1$ and  assigns $r+a$ to $r$. At the end of the for loop in \Cref{line:gbas:line7}, it assigns the estimated probability $\widehat{p}$ as $(k-1)/r$. Finally, in \Cref{line:gbas:line8}, \gbas returns the estimated probability $\widehat{p}$.

\subsection{Theoretical Analysis of Our Estimator}

\noindent The formal result of our estimator is presented below. 
\cubeprobe*

\noindent The formal proof is presented in the supplementary material.

\subsection{High-level Technical Overview}

\noindent The main idea of \gst stems from an equivalent characterization of the variation distance which states that $\dtvp(\igs, \iws) = \E_{x \sim \cD^{\igs,\psi}} (1-\cD^{\iws,\psi}(x)/\cD^{\igs,\psi}(x))$. Our goal is to estimate the ratio $\cD^{\iws,\psi}(x)/\cD^{\igs,\psi}(x)$ for some samples $x$-s drawn from $\cD^{\igs,\psi}$. As $\iws$ is known, it is sufficient to estimate $\cD^{\igs,\psi}(x)$. It is generally difficult to estimate $\cD^{\igs,\psi}(x)$. However, using self-reducibility of $\igs$ to mount subcube-conditioning access to $\cD^{\igs,\psi}$, we estimate $\cD^{\igs,\psi}(x)$ by conditioning over the $n$ conditional marginal distributions of $\cD^{\igs,\psi}$.  Using the chain formula, we obtain the value of $\cD^{\igs,\psi}(x)$ by multiplying a number of these conditional probabilities. This is achieved by the subroutine \estimate. The probability mass estimation of each conditional marginal distribution is achieved by the subroutine \gbas, which is called from \estimate. The idea of \gbas follows from \cite{rsa/Huber17}, which roughly states that to estimate the probability of head (say $p$) of a biased coin, within (multiplicative) error $\gamma_i$ and success probability at least $1-\delta$, it is sufficient to make $T$ coin tosses on average, where $T= k/p$ with $k \geq 3 \log(2/\delta)/\gamma_i^2$. The crucial parameter is the error margin $\gamma_i$ that is used in \estimate. It should be set so that after taking the errors in all the marginals into account, the total error remains bounded by the target error margin $\gamma$. Our pivotal observation is that the error distribution in the subroutine \gbas, when estimating the mass of the conditional marginal distributions, is a SubGaussian distribution (that is, a Gaussian distribution dominates its tails). Following the tail bound on the sum of SubGaussian random variables, we could afford to estimate the mass of each of the marginal with error $\gamma_i=\gamma/\sqrt{n}$ and still get an estimation of $\cD^{\igs,\psi}(x)$ with a correctness error of at most $\gamma$. That way the total sample complexity of \estimate reduces to $\widetilde{\Oh}(n /(\gamma/\sqrt{n})^2)=\widetilde{\Oh}(n^2/\gamma^2)$. As $\alpha/\gamma^2 = \mathcal{O}\left(1/\zeta^4\right)$,  we get the claimed sample complexity of \gst.

\color{black}

\section{From Estimator to Tester}
\label{sec:est2test}

\noindent We extend our design to a tester named \gstt that tests if two samplers are close or far in variation distance. 
As before, the inputs to \gstt are two self-reducible samplers $\igs,\iws$, an input $\psi$, parameters $\eps$, $\eta$, and the confidence parameter $\delta$. \gstt first computes the estimation margin-of-error $\zeta$ as ${(\eta-\eps)/2}$, and sets an intermediate confidence parameter $\delta_t$ as $2\delta$. The algorithm estimates the distance between $\igs$ and $\iws$ on input $\psi$, by invoking \gst on $\igs,\iws,\psi$ along with the estimation-margin $\zeta$ and $\delta_t$. If the computed distance $\widehat{\mathsf{dist}}$ is more than the threshold $K=(\eta + \eps)/{2}$, the tester rejects. Otherwise, the tester accepts.

\begin{algorithm}
\caption{\gstt($\igs, \iws, \psi, \eps, \eta, \delta$)}\label{balg:indentl1test}

$\zeta ={(\eta - \eps)}/{2}$ \label{line:tester:line1} \

$\delta_t = 2\delta$ \label{line:tester:line2} \

$K = {(\eta + \eps)}/{2}$
\label{line:tester:line3} \

$\widehat{\mathsf{dist}} = \gst(\igs, \iws, \psi, \zeta, \delta_t)$
\label{line:tester:line4} \

\If{$\widehat{\mathsf{dist}}  > K$}{ \label{line:tester:line5} \
\Return REJECT 
\label{line:tester:line6}  \

}

\Return ACCEPT 
\label{line:tester:line7} \
                           
\end{algorithm}

\noindent The details of \gstt are summarised below.
\begin{theo}
Consider an unknown self-reducible sampler $\igs$, a known self-reducible sampler $\iws$, an input $\psi$, closeness parameter $\eps \in (0,1)$, farness parameter $\eta \in (0,1)$ with $\eta >\eps$ and a confidence parameter $\delta \in (0,1)$. There exists a $(\eps,\eta, \delta)$-\whst \ \gstt that takes $\widetilde{\mathcal{O}}\left(n^2/\left(\eta-\eps\right)^4\right)$  samples.

\end{theo}

\noindent We note that our tester is general enough that when $\igs$ is $\eps$-close to $\iws$ in $\ell_{\infty}$-distance~\footnote{$\igs$ is $\eps$-close to $\iws$ on input $\psi$ in $\ell_{\infty}$-distance if for every $x\in \{0,1\}^n$, $(1-\eps)\cD^{\iws, \psi}(x) \leq \cD^{\igs, \psi}(x) \leq (1+ \eps) \cD^{\iws, \psi}(x)$.}, then \gstt outputs \accept. Moreover, If $\gstt$ outputs reject on input $\psi$, then one can extract a configuration (witness of rejection) $\psi_e$ such that $\igs$ and $\iws$ are $\eta$-far.

\color{black}

\section{Evaluation Results}\label{sec:evalulation-main} 

\begin{table*}[!htb]
\begin{tabular}{|l|l|lll|lll|lll|}
\hline
\multicolumn{1}{|l|}{\textbf{}}          & \multicolumn{1}{l|}{\textbf{}}          & \multicolumn{3}{l|}{~~~~{\lqck}}                                                                      & \multicolumn{3}{l|}{{~~~~~~~~~~~~~~~~~~~~\lsts}}                                                                               & \multicolumn{3}{l|}{{~~~~~~~~~~~~~~~~\lcms}}                                                                               \\ \hline
\multicolumn{1}{|l|}{{Instances}} & \multicolumn{1}{l|}{{dim}} & \multicolumn{1}{l|}{{Estd $\dtv$}} & \multicolumn{1}{l|}{{\#samples}} & \multicolumn{1}{l|}{{A/R}} & \multicolumn{1}{l|}{{Estd $\dtv$}} & \multicolumn{1}{l|}{{\#samples}} & \multicolumn{1}{l|}{{A/R}} & \multicolumn{1}{l|}{{Estd $\dtv$}} & \multicolumn{1}{l|}{{\#samples}} & \multicolumn{1}{l|}{{A/R}} \\ \hline
avgdeg\_3\_008\_2                        & 19                                      & 0.1854                            & 9986426                                 & A                                 & 0.0205                            & 11013078                                & A                                 & 0.1772                            & 9914721                                 & A                                 \\
avgdeg\_3\_010\_2                        & 30                                      & 0.1551                            & 24537279                                & A                                 & 0.0155                            & 24758147                                & A                                 & 0.1267                            & 24126731                                & A                                 \\
avgdeg\_5\_010\_3                        & 16                                      & 0.0976                            & 7593533                                 & A                                 & 0.0338                            & 7338508                                 & A                                 & 0.1135                            & 7261255                                 & A                                 \\
avgdeg\_5\_010\_4                        & 11                                      & 0.0503                            & 3486025                                 & A                                 & 0.0387                            & 3475635                                 & A                                 & 0.1147                            & 3412151                                 & A                                 \\
bn\_andes\_010\_1                        & 35                                      & 0.2742                            & 33557190                                & A                                 & 0.0396                            & 33536595                                & A                                 & 0.1601                            & 33235104                                & A                                 \\
bn\_diabetes\_010\_3                     & 26                                      & 0.1955                            & 19211200                                & A                                 & 0.0009                            & 18847561                                & A                                 & 0.1478                            & 18539480                                & A                                 \\
bn\_link\_010\_4                         & 28                                      & 0.2024                            & 21482230                                & A                                 & 0.0346                            & 22377750                                & A                                 & 0.1635                            & 21161624                                & A                                 \\
bn\_munin\_010\_1                        & 33                                      & 0.2414                            & 30348931                                & A                                 & 0.0448                            & 30693619                                & A                                 & 0.1230                            & 30218998                                & A                                 \\
bn\_pigs\_010\_1                         & 36                                      & 0.3106                            & 36917129                                & R                                 & 0.0569                            & 36311963                                & A                                 & 0.1353                            & 35978964                                & A                                 \\
bipartite\_0.2\_008\_4                   & 25                                      & 0.3204                            & 17761820                                & R                                 & 0.0073                            & 17840945                                & A                                 & 0.1153                            & 17546682                                & A                                 \\
bipartite\_0.2\_010\_1                   & 41                                      & 0.3299                            & 46244946                                & R                                 & 0.1528                            & 48135745                                & A                                 & 0.1461                            & 47003971                                & A                                 \\
bipartite\_0.5\_008\_4                   & 22                                      & 0.2977                            & 13144132                                & A                                 & 0.0528                            & 13424946                                & A                                 & 0.1059                            & 13317859                                & A                                 \\
bipartite\_0.5\_010\_1                   & 36                                      & 0.3082                            & 35875122                                & R                                 & 0.0037                            & 36728064                                & A                                 & 0.1472                            & 35823878                                & A                                 \\ \hline
\end{tabular}
\caption{For each sampler the three columns represent the estimated $\dtv$, number of samples consumed by \gst and the output of \gstt. ``A'' and ``R'' represents \accept and \reject respectively.}\label{tab:some}
\end{table*}

\noindent To evaluate the practical effectiveness of our proposed algorithms, we implemented prototype of \gst and \gstt in Python3\footnote{
codes and experimental results are available at $\mathtt{https://github.com/uddaloksarkar/cubeprobe}$.}
We use \gst to estimate the variation distance ($\dtv$) of three linear extension samplers from a perfect uniform sampler. SAT solvers power the backends of these linear extension samplers.
The objective of our empirical evaluation was to answer the following:

\noindent \textbf{RQ1} Can \gst estimate the distance of linear extension samplers from a known (e.g., uniform) sampler?

\noindent \textbf{RQ2} How many samples \gst requires to estimate the distance?

\noindent \textbf{RQ3} How do the linear extension samplers behave with an increasing number of dimensions?  

\paragraph{Boolean encoding of Poset}
Given a poset $\cP = (S, \preceq_P)$, we encode it using a Boolean formula $\varphi_\cP$ in conjunctive normal form (CNF), as described in~\cite{talvitie2018counting}:
\begin{itemize}
    \item[1] for all elements $a,b\in S$, the formula $\varphi_\cP$ contains the variables of the form $v_{ab}$ such that $v_{ab} = 1$ represents $a \preceq b$ and $v_{ab} = 0$ represents $b \preceq a$.
    \item[2] The CNF formula $\varphi_\cP$ contains the following clauses.
    Type-1: $v_{ab}$ for all $a,b\in S$ such that $a \preceq_\cP b$. This enforces the poset relation $\preceq_\cP$. 
    Type-2: $\neg v_{ab} \lor \neg v_{bc} \lor v_{ac}$ for all $a,b,c \in S$ to guarantee the transitivity.
\end{itemize} 
This reduction requires $\comb{|S|}{2}$ many variables and $\perm{|S|}{3}$ many clauses of type-2. The number of clauses of type-1 depends on the number of edges in the cover graph of $\cP$.

\vspace{-5pt}

\subsection{Experimental Setup}
\paragraph{Samplers Used:}
To assess the performance of $\gst$ and $\gstt$, we utilized three different linear extension samplers- \texttt{LxtQuicksampler}, \texttt{LxtSTS}, \texttt{LxtCMSGen}, to estimate their $\dtv$ distances from a uniform sampler. The backend of these samplers are powered by three state-of-the-art CNF samplers: \texttt{QuickSampler}~\cite{dutra2018efficient}, \texttt{STS}~\cite{ermon2012uniform}, \texttt{CMSGen}~\cite{GSCM21}. 
A poset-to-CNF encoder precedes these CNF samplers, and a Boolean string-to-poset extractor succeeds the CNF samplers to build the linear extension samplers. 
\noindent We also required access to a known uniform sampler which is equivalent to having access to a linear extension counter\footnote{For a set $\cS$ if we know the size of the set $\size{\cS}$, we know the mass of each element to be $1/\size{\cS}$ in a uniform sampler.}. We utilized an exact model counter for CNF formulas to meet this need: \texttt{SharpSAT-TD}~\cite{korhonen2021sharpsat}.

\paragraph{Poset Instances: } We adopted a subset of the poset instances from the experimental setup of \cite{talvitie2018scalable} and \cite{talvitie2018counting} to evaluate \gst and \gstt. The instances include three different kinds of posets. (a) posets of type $\mathsf{avgdeg_k}$ are generated from DAGs with average indegree of $k=3,5$; (b) posets of type $\mathsf{bipartite_p}$ have been generated by from bipartite set $S = A\cup B$ by adding the order constraint $a\prec b$ (resp. $b \prec a$) with probability $p$ (resp. $1 - p$) for all $(a,b) \in A \times B$; (c) posets of type $\mathsf{bayesiannetwork}$ is obtained from a transitive closure a randomly sampled subgraph of bayesian networks, obtained from \cite{repo}.

\vspace{-3pt}

\paragraph{Parameters Initialization:}
For our experiments with \gst, the approximation parameter $\zeta$ and confidence parameter $\delta$ are set to be 0.3 and 0.2. Our tester \gstt takes a closeness parameter $\eps$, farness parameter $\eta$, and confidence parameter $\delta$. For our experiments these are set to be $\eps : 0.01$, $\eta : 0.61$, and $\delta : 0.1$, respectively. 

\vspace{-5pt}

\paragraph{Environment} 
All experiments are carried out on a high-performance computer cluster, where each node consists of AMD EPYC 7713 CPUs with 2x64 cores and 512 GB memory. All tests were run in multi-threaded mode with 8 threads per instance per sampler with a timeout of 12 hrs. 
\begin{figure}
    \centering
    \includegraphics[scale = 0.4]{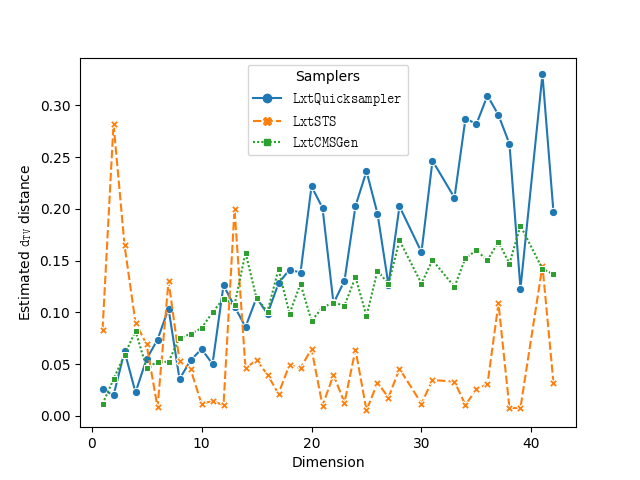}
    \caption{\small Estimated TV distances of samplers (from uniformity) as a function of the dimension. For each dimension, we take the median $\dtv$ over all the instances of that dimension.}
    \label{fig:dimdtv}
\end{figure}

\subsection{Experimental Results \& Discussion}
\paragraph{RQ1}
Table~\ref{tab:some} shows a subset of our experimental results. Due to space constraints, we have postponed presenting our comprehensive experimental results to the supplementary material. We found that among 90 instances: 
\begin{itemize}
    \item In 48 instances $\lqck$ has maximum $\dtv$, in 14 instances $\lsts$ has maximum $\dtv$ distance and in 28 instances $\lcms$ has maximum $\dtv$ distance from uniform;
    \item In 10 instances $\lqck$ has minimum $\dtv$ distance, in 69 instances $\lsts$ has minimum $\dtv$ distance and in 11 instances $\lcms$ has minimum $\dtv$ distance from uniform;
\end{itemize}
These observations indicate that $\lsts$ serves as a linear extension sampler that closely resembles uniform distribution characteristics. At the same time, $\lqck$ deviates significantly from the traits of a uniform-like linear extension sampler. $\lcms$ falls in an intermediate position between these two.

\paragraph{RQ2} 
\Cref{tab:some} reflects that the number of samples drawn by \gst depends on the dimension of an instance. Again, when the dimension is kept constant, the number of samples drawn remains similar across all runs.

\paragraph{RQ3}
In \Cref{fig:dimdtv}, we observe that for instances with lower dimensions, both $\lqck$ and $\lcms$ exhibit behavior relatively close to uniform sampling. However, as the dimension increases, $\dtv$ between these two samplers from uniformity increases. In contrast, $\lsts$ shows a different behavior. In lower dimensions, the estimated $\dtv$ distance can be notably high for certain instances, yet $\dtv$ tends to stabilize as the dimension increases. It is worth highlighting that, in higher dimensions, $\lsts$ demonstrates a more uniform-like sampling behavior compared to the other two samplers.
\section{Conclusion}
\noindent In this paper, we have designed the first self-reducible sampler tester, and used it to test linear extension samplers. We have also designed a novel variation distance estimator in the subcube-conditioning model along the way.

\paragraph{Limitations of our work}
Our algorithm takes $\widetilde{\Oh}(n^2)$ samples while the known lower bound for tolerant testing with subcube conditioning is of $\Omega(n/\log n)$ for this task \cite{canonne2020testing}. Moreover, our algorithm works when the samplers are self-reducible, which is required for our analysis. So our algorithm can not handle non-self-reducible samplers, such as in \cite{grosse2006computing,talvitie2020exact}.

\newpage

\section{Acknowledgements}
\noindent Rishiraj Bhattacharyya acknowledges the support of UKRI by EPSRC grant number EP/Y001680/1. Uddalok Sarkar is supported by the Google PhD Fellowship. Sayantan Sen's research is supported by the National Research Foundation Singapore under its NRF Fellowship Programme (NRF-NRFFAI1-2019-0002). This research is part of the programme DesCartes and is supported by the National Research Foundation, Prime Minister’s Office, Singapore, under its Campus for Research Excellence and Technological Enterprise (CREATE) programme. The computational works of this article were performed on the resources of the National Supercomputing Centre, Singapore $\mathtt{https://www.nscc.sg}$.

\bibliography{reference.bib}

\appendix

\newpage

\begin{center}
    {\Large \bf SUPPLEMENTARY MATERIAL}
\end{center}

\section{Probability Definitions and Useful Concentration Bounds}
\label{sec:prob} 
\begin{defi}[{Bernoulli distribution}]
A random variable $X \in \{0,1\}$ is said to follow Bernoulli distribution with parameter $p$ if $\Pr(X=1)=p$ and $\Pr(X=0)=1-p$ for some parameter $p \in [0,1]$. This is represented as $X \sim \mathsf{Ber}(p)$.   
\end{defi}

In our work, we use the following concentration inequalities. See~\cite{dubhashi2009concentration} for proofs.

\begin{lem}[Markov Inequality]
    Let $X$ be a random variable that only
takes non-negative values. Then for any $\alpha>0$, it holds that
\begin{align*}
    \Pr(X>\alpha) \leq \frac{\E[X]}{\alpha}
\end{align*}
\end{lem}

\begin{lem}[Chernoff-Hoeffding bound]
\label{lem:cher_bound2}
Let $X_1, \ldots, X_n$ be independent random variables such that $X_i \in [0,1]$. For $X=\sum\limits_{i=1}^n X_i$ and $\mu_l \leq \E[X] \leq \mu_h$, the followings hold for any $\delta >0$.
\begin{itemize}
\item[(i)] $\Pr \left( X \geq \mu_h + \delta \right) \leq \exp{\left(\frac{-2\delta^2}{n}\right)}$.
\item[(ii)] $\Pr \left( X \leq \mu_l - \delta \right) \leq \exp{\left(\frac{-2\delta^2}{n}\right)}$.
\end{itemize}

\end{lem}

In the analysis of our estimator, one of the crucial components is Sub-Gaussian errors which is formally defined as follows:

\begin{defi}[{\bf SubGaussian random variable}~\cite{buldygin1980sub}]
\label{lem:subg}
A random variable $X$ is said to be \emph{sub-Gaussian} (SubG in short) with parameter $\alpha^2$ if and only if its tails are dominated by a Gaussian of parameter $\alpha^2$, i.e., $$\Pr(|X| \geq t) \leq 2 \exp(-t^2/2\alpha^2), ~~~~\mbox{for all } t \geq 0.$$
\end{defi}
\begin{lem}[\cite{buldygin1980sub}]
\label{lem:subg2}
Consider $n$ independent random variables $X_i \sim \mathsf{SubG}(\alpha_i^2)$ for every $i \in [n]$, then $$X_1 + ... + X_n \sim  \mathsf{SubG}(\alpha_1^2 + \alpha_2^2 + \ldots + \alpha_n^2)$$
\end{lem}

We use the \gbas algorithm due to Huber~\cite{rsa/Huber17}. The relevant details of the algorithm is captured by the following theorem: 
\begin{theo}[Huber~\cite{rsa/Huber17}]
\label{thm:huber}
Let $\bern$ be a Bernoulli distribution parameterised by $p$. Fix $\eps\in (0,1/4)$. The randomised algorithm \gbas~\footnote{\gbas stands for Gamma Bernoulli Approximation Scheme.} on input $k \geq \frac{3\log(2\delta^{-1})}{\eps^2}$, samples $X_1,X_2,\ldots\sim \bern$ and outputs $\widehat{p}$ after $T$ samples such that
\begin{itemize}
\item[(i)] $\Pr \left(\left\lvert (\widehat{p}/p) -1 \right\rvert > \eps\right)\leq \delta$.
\item[(ii)]  $\E\left[T\right]=k/p$. 
\end{itemize}
\end{theo} 
\section{Analysis of Estimator \gst}
In this section we proof the correctness of \gst. We recall the theorem below. 
\begin{theo}\label{theo:estapp}
Consider an unknown self-reducible sampler $\igs$, a known self-reducible sampler $\iws$, an input $\psi$, an approximation parameter $\zeta \in (0,1)$ and a confidence parameter $\delta \in (0,1)$. Our \whst \ \gst is a $(\zeta,\delta)$-approx $\dtv$ estimator and \gst takes  $\widetilde{\Oh}(n^2/\zeta^4)$ samples.
\end{theo}

\subsubsection{Proof of Theorem~\ref{theo:estapp} }    

For ease of presentation, we will denote the distribution corresponding to the unknown self-reducible sampler $\igs$ on input $\psi$ as $P$ and distribution corresponding to the known self-reducible sampler $\iws$ on input $\psi$ as $Q$. We start with the correctness of \estimate.

\subsection{Correctness of \estimate}
Let $P(x)$ be the probability mass of the distribution $P$ at the string $x\in \{0,1\}^n$, and $\widehat{P}(x)$ be the estimate of $P(x)$ returned by the subroutine $\estimate$. We use  $P_i$ to denote $P({x_i}\mid {x_{1} ... x_{i-1}})$, the probability mass of the conditional marginal distribution considered at the index $i \in [n]$ at the string $x\in \{0,1\}^n$. Similarly, $\widehat{P_{i}}$ denotes the estimate of $P({x_i}\mid {x_{1} ... x_{i-1}})$ returned by the subroutine $\gbas(P\mid_{x_{1} ... x_{i-1}},x_i)$.

We define $err_i$ is the error in estimating the $i$-th marginal.
\begin{equation}\label{eqn:errori}
err_i = \size{\frac{P_i}{\widehat{P_{i}}} - 1}  
\end{equation}

\begin{cl}\label{cl:errisubg}
Errors $err_i$ are SubGaussian random variables, that is, $err_i \sim \mathsf{SubG}\left(\frac{\gamma^2}{2n \log (2n / \delta')}\right)$.
\end{cl}

\begin{proof}
 To estimate $i$-th marginal in \estimate, the subroutine \gbas takes $k \geq \frac{3n}{\gamma^2}\cdot \log\left(\frac{2n}{\delta'}\right)$ as input and estimates the $i$-th marginal. We have, from Theorem~\ref{thm:huber},
$\Pr(|err_i| > \gamma / \sqrt{n}) < \delta'/ n = 2 \exp\left(- (\gamma / \sqrt{n})^2 \cdot  \frac{n}{\gamma^2}\log (2n / \delta')\right)$.
So, $err_i \sim  \mathsf{SubG}\left(\frac{\gamma^2}{2 n \log (2n / \delta')}\right)$.
\end{proof}

\begin{cl}\label{cl:bounderr}
Consider the random variable $err=\sum_{i=1}^n err_i$. Then we have $\Pr(|err| > \gamma) \leq \frac{\delta'}{n}$.    
\end{cl}

\begin{proof}
From Lemma~\ref{lem:subg2}, the total error in estimating is
$err = \sum_i err_i \sim \mathsf{SubG}\left(\frac{\gamma^2}{2 \log (2n / \delta')}\right)$. By \Cref{lem:subg}, we conclude that $\Pr(|err| > \gamma) \leq 2\exp(-\frac{\gamma^2 \cdot 2 \log (2n/\delta')}{\gamma^2}) = \frac{\delta'}{n}$.
\end{proof}
\noindent The following lemma argues the correctness of \estimate.
\begin{lem}[Correctness of \estimate]\label{bcl:estimatel1}
Fix $\gamma, \delta \in (0,1)$. Consider $\delta'= \frac{\delta}{2\alpha}$ for $\alpha=\frac{1}{2\zeta^2} \log \frac{4}{\delta}$, as fixed in the algorithms.
\estimate($P,x, \gamma/1.11,\delta'$) estimates $P(x)$ within a multiplicative factor of $(1 \pm \gamma)$ with probability at least $(1 - \delta')$, where the probability is taken over the internal randomness of \estimate. Moreover, the expected number of samples required by \estimate is $\Oh\left(\frac{n^2}{\gamma^2} \cdot \log\left(\frac{2n}{\delta'}\right)\right)$.
\end{lem}

\begin{proof}

From the correctness of \gbas, we know that with probability at least $\delta'/n$, the following holds for every $i \in [n]$:
\begin{equation}\label{eqn:estPxi}
\size{\widehat{P_{i}}- P(x_i | x_1...x_{i-1})} \leq \frac{\gamma}{\sqrt{n}} \cdot P(x_i | x_1...x_{i-1})
\end{equation}
Therefore, from Equation~\ref{eqn:errori}, $err_i > \frac{\gamma}{\sqrt{n}}$. Taking product of above inequalities over all $i\in[n]$
\begin{equation}\label{beqn:estPxexact}
\Pi_{i=1}^n\left(1-err_i\right) P(x) \leq \widehat{P}(x) \leq \Pi_{i=1}^n\left(1+err_i\right) P(x)    
\end{equation}

Assuming $err_i \in (0,0.8)$, we know that $1-err_i \geq e^{-1.11\cdot err_i}$~\footnote{We will be using $1.11$ as the exponent in the approximation as compared to the more commonly used exponent $2$. This is done in order to obtain better sample complexity in our experiments.}. Moreover, $(1+x) \leq e^x$. Thus from the above expression, we obtain a multiplicative estimate of the probability mass $P(x)$ as follows:
\begin{equation*}
e^{-\sum_{i=1}^n 1.11 \cdot err_i} P(x) \leq \widehat{P}(x) \leq e^{\sum_{i=1}^n err_i} P(x)    
\end{equation*}

Since $\sum_{i=1}^n err_i=err$, we can write the above as 
\begin{equation}\label{beqn:estPx}
e^{- 1.11 \cdot err} P(x) \leq \widehat{P}(x) \leq e^{err} P(x)    
\end{equation}

Now note that in the interval $x \in (0,1)$, $e^{x} \leq (1+1.11x)$ and $1-1.11x \leq e^{-1.11x}$. From \Cref{cl:bounderr}, we know that $err \leq \gamma$ with high probability. Thus we obtain the following,
\begin{equation}\label{beqn:estP'x}
(1-1.11\gamma) P(x) \leq \widehat{P}(x) \leq (1+1.11\gamma) P(x)    
\end{equation}

Since we are calling $\estimate(P,x,\gamma/1.11,\delta')$, combining the above with \Cref{beqn:estPx}, we obtain a multiplicative estimate of the probability mass $P(x)$ as follows:
\begin{equation}\label{beqn:estP'xapprox}
(1-\gamma) P(x) \leq \widehat{P}(x) \leq (1+\gamma) P(x)    
\end{equation}

To establish correctness, recall that for any index $i \in [n]$, \Cref{eqn:estPxi} fails to hold with probability at most $\delta'/n$. Taking union bound, the error in estimating one of the $n$ marginals in \Cref{beqn:estPx} is at most $\delta'$.

Finally we establish the sample complexity of \estimate. To do this, we first determine the number of samples drawn by \gbas to estimate one marginal. Suppose $\widehat{P_{i}}$ denotes the probability mass $P_i = P(x_i \mid x_1...x_{i-1})$. From Theorem~\ref{thm:huber}, the expected number of samples taken by \gbas is given by
\begin{align*}
  &\E \left(\frac{3 \times 1.232 n}{\gamma^2} \cdot \log \left(\frac{2n}{\delta'}\right)\cdot\frac{1}{\widehat{P_{i}}}\right)\\
  \leq& \E \left(\frac{4 n}{\gamma^2} \cdot \log \left(\frac{2n}{\delta'}\right)\cdot\frac{1}{\widehat{P_{i}}}\right)\\ 
  =& \frac{4 n}{\gamma^2} \cdot \log \left(\frac{2n}{\delta'}\right)\left(\widehat{P}_{x_i=0}\cdot\frac{1}{\widehat{P}_{x_i=0}}+ \widehat{P}_{x_i=1}\cdot\frac{1}{\widehat{P}_{x_i=1}}\right)\\
  =& \frac{8n}{\gamma^2} \cdot \log \left(\frac{2n}{\delta'}\right)
\end{align*}

As the \estimate calls the subroutine \gbas $n$ times, therefore the total expected sample complexity is given by 
\begin{align*}
\frac{8 n^2}{\gamma^2} \cdot \log \left(\frac{2n}{\delta'}\right)    
\end{align*}
This concludes the proof of the claim.

\end{proof}

\begin{cl}\label{bclaim:estclosel1}
Let $\widehat{P}(x)$ be the estimate of $P(x)$ such that $(1-\gamma)P(x) \leq \widehat{P}(x) \leq (1+\gamma)P(x)$ holds. Then $\E\limits_{x \sim P}\left[\size{\frac{Q(x)}{P(x)} - \frac{Q(x)}{\widehat{P}(x)}}\right] \leq \frac{\gamma}{1 - \gamma}$.
\end{cl}

\begin{proof}
We know that $(1-\gamma)P(x) \leq \widehat{P}(x) \leq (1+\gamma)P(x)$. Hence,
\begin{eqnarray*}
 \frac{1}{(1+\gamma)P(x)} \leq \frac{1}{\widehat{P}(x)} \leq \frac{1}{(1-\gamma)P(x)} 
\end{eqnarray*}

 Thus we have, 
\begin{eqnarray*}
\frac{1}{P(x)} - \frac{1}{(1-\gamma)P(x)} 
& \leq & \frac{1}{P(x)} -  \frac{1}{\widehat{P}(x)} \\
&\leq & \frac{1}{P(x)} - \frac{1}{(1+\gamma)P(x)} 
\end{eqnarray*}

Therefore, 
\begin{eqnarray*}
&&\size{\frac{1}{P(x)} -  \frac{1}{\widehat{P}(x)}} \\
&\leq & \max \left\{\frac{\gamma}{(1-\gamma) P(x)}, \frac{\gamma}{(1+ \gamma)P(x)}\right\} \\
&= &\frac{\gamma}{(1-\gamma) P(x)}
\end{eqnarray*}

Therefore, by multiplying $Q(x)$ on both sides and taking expectation we have, 
\begin{eqnarray*}
\E\limits_{x \sim P}\left[\size{\frac{Q(x)}{P(x)} -  \frac{Q(x)}{\widehat{P}(x)}}\right] \leq \frac{\gamma}{1-\gamma}\cdot \E\limits_{x \sim P}\left[\frac{Q(x)}{P(x)}\right] = \frac{\gamma}{1-\gamma}.
\end{eqnarray*}

This completes the proof of the claim.
\end{proof}

\begin{lem}\label{lem:dtvexpl1}
$\dtv(P,Q) = \E\limits_{x \sim P}\max \left( 0, 1- \frac{Q(x)}{P(x)} \right) = \E\limits_{y \sim Q}\max \left( 0, 1- \frac{P(y)}{Q(y)} \right)$.    
\end{lem}

\begin{proof}
From the definitions we have
\begin{eqnarray*}
\dtv(P,Q) &=& \sum_{x : P(x) > Q(x)} (P(x) - Q(x)) \\ &=& \sum_{x : P(x) > Q(x)} P(x) \left(1 - \frac{Q(x)}{P(x)}\right) \\ &=& \sum_{x} P(x) \max \left\{0, 1 - \frac{Q(x)}{P(x)}\right\}\\ 
&=& \E_{x \sim P} \max \left\{0, 1- \frac{Q(x)}{P(x)}\right\} 
\end{eqnarray*}
Similarly, we can also say that $ \dtv(P,Q) = \E\limits_{y \sim Q}\max \left( 0, 1- \frac{P(y)}{Q(y)} \right)$.
\end{proof}

\subsection{Correctness of \gst}

\noindent Now we are ready to prove the correctness of \gst. 
\begin{lem}
Let $v$ be the output of \gst($\igs, \iws, \psi, \zeta, \delta$). It holds, with probability at least 1-$\delta$ that
    $$ \dtvp(\igs,\iws) - \zeta \leq v \leq \dtvp(\igs,\iws) + \zeta $$
\end{lem}

\begin{proof}
From the description of Algorithm \gst (\Cref{balg:indentl1}), $v=val/\alpha$, where $val = \sum_{i=1}^\alpha \max\left(0, 1 - \frac{Q(x)}{\widehat{P}(x)}\right)$. Recall, $\widehat{P}(x)$ denotes the estimate of $P(x)$ as returned by $\estimate$. 

 Consider the event $\cA$:
$$\cA:=\estimate \ \mbox{outputs a correct estimate of} \ P(x)$$
From \Cref{bcl:estimatel1}, we know that $\Pr(\cA)=1-\delta'$, where $\delta'=\delta/2\alpha$. From now on, we work in the  conditional space that the even $\cA$ holds.

To prove the lemma, we need to prove the following:
$$ \dtvp(\igs,\iws) - \zeta \leq \frac{val}{\alpha} \leq \dtvp(\igs,\iws) + \zeta $$

\textbf{Case 1 } Suppose $Q(x) < P(x) < \widehat{P}(x)$. Then we have $\frac{Q(x)}{\widehat{P}(x)} < \frac{Q(x)}{P(x)} < 1$. Therefore using the triangle inequality, we can say the following: 
\begin{eqnarray*}
 \size{1 - \frac{Q(x)}{\widehat{P}(x)}} &\leq& \size{1 - \frac{Q(x)}{P(x)}} + \size{\frac{Q(x)}{P(x)} - \frac{Q(x)}{\widehat{P}(x)}} \\
  1 - \frac{Q(x)}{\widehat{P}(x)} &\leq& 1 - \frac{Q(x)}{P(x)} + \size{\frac{Q(x)}{P(x)} - \frac{Q(x)}{\widehat{P}(x)}} \\
 \max\left(0, 1 - \frac{Q(x)}{\widehat{P}(x)}\right) &\leq& \max\left(0, 1 - \frac{Q(x)}{P(x)}\right) \\ \hspace{50pt} \ &+& \size{\frac{Q(x)}{P(x)} - \frac{Q(x)}{\widehat{P}(x)}}
\end{eqnarray*}  

\textbf{Case 2 } Suppose $Q(x) < \widehat{P}(x) < P(x)$. Then we have $\frac{Q(x)}{P(x)} < \frac{Q(x)}{\widehat{P}(x)} < 1$. Therefore using the same triangle inequality we have: 
\begin{eqnarray*}
 \size{1 - \frac{Q(x)}{\widehat{P}(x)}} &\leq& \size{1 - \frac{Q(x)}{P(x)}} + \size{\frac{Q(x)}{P(x)} - \frac{Q(x)}{\widehat{P}(x)}} \\
  1 - \frac{Q(x)}{\widehat{P}(x)} &\leq& 1 - \frac{Q(x)}{P(x)} + \size{\frac{Q(x)}{P(x)} - \frac{Q(x)}{\widehat{P}(x)}} \\
 \max\left(0, 1 - \frac{Q(x)}{\widehat{P}(x)}\right) &\leq& \max\left(0, 1 - \frac{Q(x)}{P(x)}\right) \\ &+& \size{\frac{Q(x)}{P(x)} - \frac{Q(x)}{\widehat{P}(x)}}
\end{eqnarray*}  

\textbf{Case 3 } Suppose $\widehat{P}(x) < P(x) < Q(x)$ . Then using the facts that $\frac{Q(x)}{\widehat{P}(x)} > 1$ and $\frac{Q(x)}{P(x)} > 1$ and $\size{\frac{Q(x)}{P(x)} - \frac{Q(x)}{\widehat{P}(x)}} \geq 0$, we have:
\begin{eqnarray*}
  \max\left(0, 1 - \frac{Q(x)}{\widehat{P}(x)}\right) &\leq& \max\left(0, 1 - \frac{Q(x)}{P(x)}\right) \\ &+& \size{\frac{Q(x)}{P(x)} - \frac{Q(x)}{\widehat{P}(x)}}
\end{eqnarray*}  

\textbf{Case 4 } Suppose $P(x) < \widehat{P}(x) < Q(x)$ . Follows from the above case.

\textbf{Case 5 } Suppose $P(x) < Q(x) < \widehat{P}(x)$. Then using the fact $\frac{Q(x)}{P(x)} > 1$ and $\frac{Q(x)}{\widehat{P}(x)} < 1$, we have: 
\begin{eqnarray*}
\size{1 - \frac{Q(x)}{\widehat{P}(x)}} &\leq& \size{\frac{Q(x)}{P(x)} - \frac{Q(x)}{\widehat{P}(x)}} \\
1 - \frac{Q(x)}{\widehat{P}(x)} &\leq& 0 + \size{\frac{Q(x)}{P(x)} - \frac{Q(x)}{\widehat{P}(x)}} \\
 \max\left(0, 1 - \frac{Q(x)}{\widehat{P}(x)}\right) &\leq& \max\left(0, 1 - \frac{Q(x)}{P(x)}\right) \\ &+& \size{\frac{Q(x)}{P(x)} - \frac{Q(x)}{\widehat{P}(x)}}
\end{eqnarray*}  

\textbf{Case 6 } Suppose $\widehat{P}(x) < Q(x) < P(x)$. Then using the fact $\frac{Q(x)}{P(x)} < 1$ and $\frac{Q(x)}{\widehat{P}(x)} > 1$, we can say that:
\begin{eqnarray*}
0 &\leq& \left(1 - \frac{Q(x)}{P(x)}\right) + \size{\frac{Q(x)}{P(x)} - \frac{Q(x)}{\widehat{P}(x)}} 
\end{eqnarray*}
\begin{eqnarray*}
\max\left(0, 1 - \frac{Q(x)}{\widehat{P}(x)}\right) 
&\leq& \max\left(0, 1 - \frac{Q(x)}{P(x)}\right) \\ &+& \size{\frac{Q(x)}{P(x)} - \frac{Q(x)}{\widehat{P}(x)}}
\end{eqnarray*}  

Therefore for all $x$ we have,
\begin{eqnarray*}
\max\left(0, 1 - \frac{Q(x)}{\widehat{P}(x)}\right) &\leq& \max\left(0, 1 - \frac{Q(x)}{P(x)}\right) \\ &+& \size{\frac{Q(x)}{P(x)} - \frac{Q(x)}{\widehat{P}(x)}} 
\end{eqnarray*}

Similarly, we also can show the following:
\begin{eqnarray*}
\max\left(0, 1 - \frac{Q(x)}{P(x)}\right) &-& \size{\frac{Q(x)}{P(x)} - \frac{Q(x)}{\widehat{P}(x)}} \\ &\leq& \max\left(0, 1 - \frac{Q(x)}{\widehat{P}(x)}\right)
\end{eqnarray*}

Therefore taking expectation with respect to $x$ sampled from $P$ and using \Cref{bclaim:estclosel1}, we have the following inequality:
\begin{eqnarray*}
\dtv(P, Q) - \frac{\gamma}{1-\gamma} &\leq& \E\limits_{x \sim P}\left[\max\left(0, 1 - \frac{Q(x)}{\widehat{P}(x)}\right)\right] \\ &\leq& \dtv(P, Q) + \frac{\gamma}{1-\gamma}    \\
\dtv(P, Q) - \frac{\zeta}{2} &\leq& \E\limits_{x \sim P}\left[\max\left(0, 1 - \frac{Q(x)}{\widehat{P}(x)}\right)\right] \\ &\leq& \dtv(P, Q) + \frac{\zeta}{2} 
\end{eqnarray*}

the last inequality follows from $\zeta = \frac{2\gamma}{1 - \gamma}$. Recall that $\alpha =\frac{2}{ \zeta^2} \log \frac{4}{\delta} $. Therefore using Chernoff-Hoeffding bound we have, 
\begin{eqnarray*}
    \Pr\left(\frac{val}{\alpha} > \dtv(P, Q) + \zeta \right) 
    &\leq& \exp\left(-\frac{\zeta^2 \alpha^2 }{2\alpha}\right) \\ 
    &\leq& \frac{\delta}{4}
\end{eqnarray*}

Similarly, we can prove the following:
$$\Pr \left(\frac{val}{\alpha} < \dtv(P, Q) - \zeta\right) \leq \frac{\delta}{4}$$

Thus, under the condition that the event $\cA$ holds, combining the above, we conclude that our algorithm outputs $\frac{val}{\alpha} \in [\dtv(P, Q) - \zeta, \dtv(P, Q) + \zeta]$ with probability at least $1- \frac{\delta}{2}$. 

What is left is to bound the probability that $\cA$ does not hold for some $x\in S$. Recall, from \Cref{bcl:estimatel1} for each $x$, it holds that $\Pr(\neg\cA) \leq \delta/2\alpha$. As the subroutine \estimate is called for $\alpha$ times, the probability that $\cA$ does not hold for some $x\in S$ is at most $\delta/2$. Thus the total error of algorithm \gst is bounded by $\delta$.  
\end{proof}

\subsection{Sample Complexity}

\begin{lem}\label{lem:sampleest}
The total sample complexity of \gst is $\Oh\left(\frac{n^2 (2+\zeta)^2}{\zeta^4} \cdot \log \left(\frac{4n \alpha}{\delta}\right) \log \frac{4}{\delta} \right)$ where $\alpha=\frac{1}{2\zeta^2} \log \frac{4}{\delta}$.
\end{lem}

\begin{proof}
From the description of the algorithm, we know $\alpha=\frac{1}{2\zeta^2} \log \frac{4}{\delta}$. From \Cref{bcl:estimatel1}, we know that \estimate requires $\Oh\left(\frac{n^2}{\gamma^2} \cdot \log\left(\frac{2n}{\delta'}\right)\right)$ samples on expectation in every iteration, where $\delta'=\delta/2 \alpha$ and $\gamma=\zeta/1.11(2+ \zeta)$. Since \gst runs $\alpha$ iterations, using Markov inequality, the total sample complexity of \gst is $\Oh\left(\frac{n^2 (2+\zeta)^2}{\zeta^4} \cdot \log \left(\frac{4n \alpha}{\delta}\right) \log \frac{4}{\delta} \right)$.  

\end{proof}

\section{Correctness of Tester}
Similar to the analysis of the estimator, for ease of presentation, we will denote the distribution corresponding to the unknown self-reducible sampler $\igs$ on input $\psi$ as $P$ and distribution corresponding to the known self-reducible sampler $\iws$ on input $\psi$ as $Q$.

\begin{theo}
Consider an unknown self-reducible sampler $\igs$, a known self-reducible sampler $\iws$, an input $\psi$, closeness parameter $\eps \in (0,1)$, farness parameter $\eta \in (0,1)$ with $\eta >\eps$ and a confidence parameter $\delta \in (0,1)$. There exists a $(\eps,\eta, \delta)$-\whst \ \gstt that takes $\widetilde{\mathcal{O}}\left(n^2/\left(\eta-\eps\right)^4\right)$  samples.
\end{theo}

We will prove the correctness of \gstt into two parts:
\begin{itemize}
    \item[1] \emph{Completeness:} If $\cI_{\cG}$ is $\eps$-close to $\iws$ on input $\psi$, \gstt outputs \accept.
    
    \item[2] \emph{Soundness:} if $\cI_{\cG}$ is $\eta$-far from $\iws$ on input $\psi$, \gstt outputs \reject.
\end{itemize}

\subsection*{Proof of Completeness}
\begin{lem}\label{lem:l1indcomp}
Suppose $\igs$ is $\eps$-close to the known self-reducible sampler $\iws$ on input $\psi$. Then \gstt outputs \accept with probability at least $1- \delta$.

\end{lem}

\begin{proof}[Proof of Lemma~\ref{lem:l1indcomp}]

Let us first consider the event $\cA$:
$$\cA:=\estimate \ \mbox{outputs a correct estimate of P(x)}$$

From \Cref{bcl:estimatel1}, we know that $\Pr(\cA)=1-\delta'$, where $\delta'=\delta/2\alpha$. Let us now work on the conditional space that event $\cA$ holds.

Since $\dtv(P,Q)\leq \eps$, from Lemma~\ref{lem:dtvexpl1} we know that $$\E\limits_{x\sim P}\left[\max\left(0, 1 - \frac{Q(x)}{P(x)}\right)\right] \leq \eps$$

Therefore following the same arguments as in \Cref{lem:dtvexpl1} we have $$\E\limits_{x\sim P}\left[\max\left(0, 1 - \frac{Q(x)}{\widehat{P}(x)}\right)\right] \leq \eps + \frac{\gamma}{1-\gamma}$$

Recall that $\frac{2\gamma}{1 - \gamma}=\frac{\eta-\eps}{2}$. 

We derive
\begin{eqnarray*}
\E\limits_{x \sim P}\left[\max\left(0, 1- \frac{Q(x)}{\widehat{P}(x)}\right)\right] &\leq& \eps + \frac{(\eta-\eps)}{4}=\frac{\eta+3\eps}{4} 
\end{eqnarray*}

Recall that $K = \frac{\eta+\eps}{2}$ and $\alpha = \frac{8}{(\eta - \eps)^2} \log \left(\frac{2}{\delta}\right)$. Using the Chernoff-Hoeffding bound we have,
\begin{eqnarray*}
    \Pr\left(\frac{val}{\alpha} > K \right) 
    &=& \Pr\left(\frac{val}{\alpha} > \frac{\eta + 3\eps}{4} + \frac{\eta - \eps}{4}  \right) \\
    &\leq& \exp\left(-2\left(\frac{\eta - \eps}{4}\right)^2 \cdot \alpha\right) \\ 
    &\leq& \frac{\delta}{2}
\end{eqnarray*}

Thus, under the conditional space that the event $\cA$ holds, our algorithm accepts with probability at least $1- \frac{\delta}{2}$. Since $\Pr(\cA) \geq 1- \delta'$ with $\delta'=\delta/2\alpha$ and the subroutine \estimate is called for $\alpha$ times, the total error following \estimate is bounded by $\delta/2$. of our algorithm is bounded by $\alpha \cdot \delta' + \frac{\delta}{2}$. Thus the total error of our algorithm \gst is bounded by $\delta$. So we conclude that when $\dtv(P,Q) \leq \eps$, \gst accepts with probability at least $1- \delta$.

\end{proof}

\subsection*{Proof of Soundness}

\begin{lem}
Suppose $\igs$ is $\eta$-far from the known self-reducible sampler $\iws$ on input $\psi$. Then \gst outputs \reject with probability at least $1- \delta$.
\end{lem}

\begin{proof}
Like the completeness proof (\Cref{lem:l1indcomp}), we first work on the conditional space that the event $\cA$ holds.

From the assumption, we know that $P$ and $Q$ are $\eta$-far, that is, $\dtv(P,Q) \geq \eta$. From \Cref{lem:dtvexpl1}, we know that
$$\E\limits_{x\sim P}\left[\max\left(0, 1 - \frac{Q(x)}{P(x)}\right)\right] \geq \eta$$

Therefore following the same arguments as in \Cref{lem:dtvexpl1} we have $$\E\limits_{x\sim P}\left[\max\left(0, 1 - \frac{Q(x)}{\widehat{P}(x)}\right)\right] \geq \eta - \frac{\gamma}{1-\gamma}$$

Recall that $\frac{2\gamma}{1 - \gamma}=\frac{\eta-\eps}{2}$. 

We derive
\begin{eqnarray*}
\E\limits_{x \sim P}\left[\max\left(0, 1- \frac{Q(x)}{\widehat{P}(x)}\right)\right] &\geq& \eta - \frac{(\eta-\eps)}{4}=\frac{3\eta+\eps}{4} 
\end{eqnarray*}

Recall that $K = \frac{\eta+\eps}{2}$ and $\alpha = \frac{8}{(\eta - \eps)^2} \log \left(\frac{2}{\delta}\right)$. Using the Chernoff-Hoeffding bound we have,
\begin{eqnarray*}
    \Pr\left(\frac{val}{\alpha} < K \right) 
    &=& \Pr\left(\frac{val}{\alpha} < \frac{3\eta + \eps}{4} - \frac{\eta - \eps}{4}  \right) \\
    &\leq& \exp\left(-2\left(\frac{\eta - \eps}{4}\right)^2 \cdot \alpha\right) \\ 
    &\leq& \frac{\delta}{2}
\end{eqnarray*}

Thus, under the conditional space that the event $\cA$ holds, our algorithm rejects with probability at least $1- \frac{\delta}{2}$. Since $\Pr(\cA) \geq 1- \delta'$ with $\delta'=\delta/2\alpha$ and the subroutine \estimate is called for $\alpha$ times, the total error following \estimate is bounded by $\delta/2$. of our algorithm is bounded by $\alpha \cdot \delta' + \frac{\delta}{2}$. Thus the total error of our algorithm \gst is bounded by $\delta$. So we conclude that when $\dtv(P,Q) \geq \eta$, \gst rejects with probability at least $1- \delta$.

\end{proof}

The sample complexity bound of \gstt follows from \Cref{lem:sampleest}.
\section{Extended Experimental Results}

\noindent In this section we describe the extended experimental results of \gst and \gstt. As mentioned in the main paper, for each instance-sampler pair, 8 cores have been employed. All the experiments have been carried out in a high-performance cluster, where each node consists of AMD EPYC 7713 CPUs with $2 \times 64$ cores and 512 GB memory. Here we provide two sets of experimental results with different parameter settings. 
\begin{itemize}
    \item[1] \emph{Experiment-1}: \Cref{tab:all-1} and \Cref{tab:all-1b} show the outcomes of Experiment-1. For the estimator \gst, the approximation parameter $\zeta$ is set to $0.3$, and for the tester $\gstt$ the tolerance parameter $\eps$ and intolerance parameter $\eta$ are set to $0.005$ and $0.605$ respectively. The experiment was run with a cutoff time of 10 hrs. The experiment ended for all instances and for all the samplers within the time-limit. 
    \item[2] \emph{Experiment-2}:  \Cref{tab:all-2a} and \Cref{tab:all-2b} show the outcomes of Experiment-2. For the estimator \gst, the approximation parameter $\zeta$ is set to $0.2$, and for the tester $\gstt$ the tolerance parameter $\eps$ and intolerance parameter $\eta$ are set to $0.005$ and $0.405$ respectively. The experiment was run with a cutoff time of 24 hours. Notably, with regard to \lcms, the experiment concluded for all but one instance. Concerning \lqck~ and \lsts, certain instances exceeded the time limit, denoted by "TLE" in Table \ref{tab:all-2a} and Table \ref{tab:all-2b}.
\end{itemize}

\begin{figure}[!htb]
    \centering
    \begin{subfigure}[b]{0.3\textwidth}
         \centering
         \includegraphics[scale = 0.4]{dimvsdtv-1.png}
         \caption{~~~~~{Experiment-1}}
         \label{fig:dimdtv1}
    \end{subfigure}

    \begin{subfigure}[b]{0.3\textwidth}
         \centering
         \includegraphics[scale = 0.4]{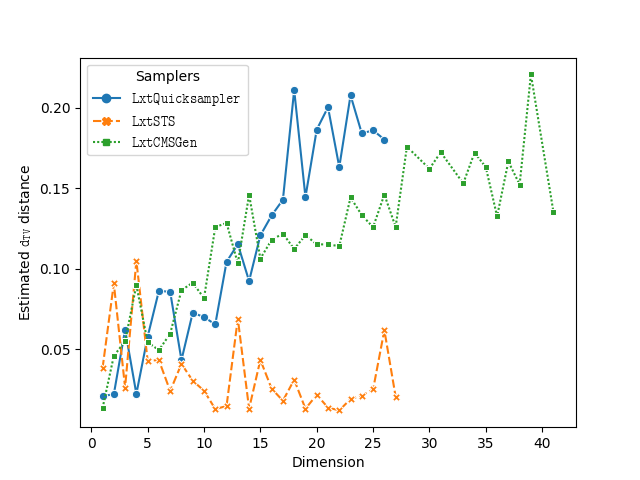}         
         \caption{~~~~~{Experiment-2}}
         \label{fig:dimdtv2}
    \end{subfigure}
    
    \caption{Estimated TV distance of samplers as a function of dimension. For each dimension, we take the median $\dtv$ over all the instances of that dimension.}
    \label{fig:dimdtv-all}
\end{figure}

\vspace{-1pt}
The trends observed in both Figure \ref{fig:dimdtv1} and Figure \ref{fig:dimdtv2} consistently highlight the fact that \lqck\ is significantly distant from uniform behavior and tends to perform even less effectively in scenarios involving higher dimensions. Conversely, \lsts\ demonstrates relatively satisfactory performance even as the dimensions increase.

\begin{table*}[t]
\begin{tabular}{|l|l|lll|lll|lll|}
\hline
\multicolumn{1}{|l|}{\textbf{}}          & \multicolumn{1}{l|}{}                   & \multicolumn{3}{l|}{~~~~~~~~~~~~~~\texttt{LxtQuickSampler}}                                                                             & \multicolumn{3}{l|}{~~~~~~~~~~~~~~~~~~\texttt{LxtSTS}}                                                                                & \multicolumn{3}{l|}{~~~~~~\texttt{LxtCMSGen}}                                                                       \\ \hline
\multicolumn{1}{|l|}{Instances} & \multicolumn{1}{l|}{dim} & \multicolumn{1}{l|}{Est $\dtv$} & \multicolumn{1}{l|}{\#samples} & \multicolumn{1}{l|}{A / R} & \multicolumn{1}{l|}{Est $\dtv$} & \multicolumn{1}{l|}{\#samples} & \multicolumn{1}{l|}{A / R} & \multicolumn{1}{l|}{Est $\dtv$} & \multicolumn{1}{l|}{\#samples} & \multicolumn{1}{l|}{A / R} \\ \hline
avgdeg\_3\_008\_0                        & 16                                      & 0.0770                            & 6791949                                 & A                                 & 0.0396                            & 6799423                                 & A                                 & 0.1459                            & 6755638                                 & A                                 \\
avgdeg\_3\_008\_1                        & 15                                      & 0.1603                            & 6235899                                 & A                                 & 0.0987                            & 6384883                                 & A                                 & 0.0916                            & 6262721                                 & A                                 \\
\hline
\end{tabular}
\caption{Extended table of Experiment-1: for each sampler the three columns represent the estimated $\dtv$, number of samples consumed by \gst, and the output of \gstt. ``A'' and ``R'' represents \accept and \reject respectively.}\label{tab:all-1}
\end{table*}

\begin{table*}
\begin{tabular}{|l|l|lll|lll|lll|}
\hline
\multicolumn{1}{|l|}{\textbf{}}          & \multicolumn{1}{l|}{}                   & \multicolumn{3}{l|}{~~~~~~~~~~~~~~\texttt{LxtQuickSampler}}                                                                             & \multicolumn{3}{l|}{~~~~~~~~~~~~~~~~~~\texttt{LxtSTS}}                                                                                & \multicolumn{3}{l|}{~~~~~~\texttt{LxtCMSGen}}                                                                       \\ \hline
\multicolumn{1}{|l|}{Instances} & \multicolumn{1}{l|}{dim} & \multicolumn{1}{l|}{Est $\dtv$} & \multicolumn{1}{l|}{\#samples} & \multicolumn{1}{l|}{A / R} & \multicolumn{1}{l|}{Est $\dtv$} & \multicolumn{1}{l|}{\#samples} & \multicolumn{1}{l|}{A / R} & \multicolumn{1}{l|}{Est $\dtv$} & \multicolumn{1}{l|}{\#samples} & \multicolumn{1}{l|}{A / R} \\ \hline
bn\_link\_010\_4                         & 28                                      & 0.2024                            & 21482230                                & A                                 & 0.0346                            & 22377750                                & A                                 & 0.1635                            & 21161624                                & A                                 \\
bn\_munin\_008\_0                        & 20                                      & 0.1772                            & 11426488                                & A                                 & 0.0831                            & 11433681                                & A                                 & 0.0741                            & 11279783                                & A                                 \\
\hline
\end{tabular}
\caption{Extended table of Experiment-1: for each sampler the three columns represent the estimated $\dtv$, number of samples consumed by \gst, and the output of \gstt. ``A'' and ``R'' represents \accept and \reject respectively.}\label{tab:all-1b}
\end{table*}

\begin{table*}[]
\begin{tabular}{|l|l|lll|lll|lll|}
\hline
\multicolumn{1}{|l|}{}            & \multicolumn{1}{l|}{}    & \multicolumn{3}{l|}{\lqck}                                                 & \multicolumn{3}{l|}{\lsts}                                                          & \multicolumn{3}{l|}{\lcms}                                                       \\ \hline
\multicolumn{1}{|l|}{Instances}   & \multicolumn{1}{l|}{dim} & \multicolumn{1}{l|}{Est dTV} & \multicolumn{1}{l|}{\#samples} & \multicolumn{1}{l|}{A / R} & \multicolumn{1}{l|}{Est dTV} & \multicolumn{1}{l|}{\#samples} & \multicolumn{1}{l|}{A / R} & \multicolumn{1}{l|}{Est dTV} & \multicolumn{1}{l|}{\#samples} & \multicolumn{1}{l|}{A / R} \\ \hline
avgdeg\_3\_008\_0                 & 16                       & 0.1039                       & 57780610                       & A                          & 0.0138                       & 58313012                       & A                          & 0.1878                       & 56107260                       & A                          \\
avgdeg\_3\_008\_1                 & 15                       & 0.1131                       & 53251513                       & A                          & 0.0538                       & 53541132                       & A                          & 0.0983                       & 53016623                       & A                          \\
avgdeg\_3\_008\_2                 & 19                       & 0.1677                       & 85093433                       & A                          & 0.0149                       & 93032499                       & A                          & 0.2069                       & 82702207                       & R                          \\
\hline
\end{tabular}
\caption{Table of Experiment-2: for each sampler the three columns represent the estimated $\dtv$, number of samples consumed by \gst, and the output of \gstt. ``A'' and ``R'' represents \accept and \reject respectively. ``TLE'' indicates ``Time limit exceeded''.}\label{tab:all-2a}
\end{table*}

\begin{table*}
\begin{tabular}{|l|l|lll|lll|lll|}
\hline
\multicolumn{1}{|l|}{\textbf{}}          & \multicolumn{1}{l|}{}                   & \multicolumn{3}{l|}{~~~~~~~~~~~~~~\texttt{LxtQuickSampler}}                                                                             & \multicolumn{3}{l|}{~~~~~~~~~~~~~~~~~~\texttt{LxtSTS}}                                                                                & \multicolumn{3}{l|}{~~~~~~\texttt{LxtCMSGen}}                                                                       \\ \hline
\multicolumn{1}{|l|}{Instances} & \multicolumn{1}{l|}{dim} & \multicolumn{1}{l|}{Est $\dtv$} & \multicolumn{1}{l|}{\#samples} & \multicolumn{1}{l|}{A / R} & \multicolumn{1}{l|}{Est $\dtv$} & \multicolumn{1}{l|}{\#samples} & \multicolumn{1}{l|}{A / R} & \multicolumn{1}{l|}{Est $\dtv$} & \multicolumn{1}{l|}{\#samples} & \multicolumn{1}{l|}{A / R} \\ \hline
bn\_munin\_008\_0    & 20                       & TLE                          &                                & TLE                        & 0.0290                       & 97911223                       & A                          & 0.1058                       & 95117302                       & A                          \\
bn\_munin\_008\_1    & 17                       & 0.1201                       & 68388967                       & A                          & 0.0133                       & 70664348                       & A                          & 0.1218                       & 68691968                       & A                          \\
bn\_munin\_008\_2    & 7                        & 0.0856                       & 13039482                       & A                          & 0.0240                       & 12770555                       & A                          & 0.0596                       & 12703958                       & A                          \\
\hline
\end{tabular}
\caption{Table of Experiment-2: for each sampler the three columns represent the estimated $\dtv$, number of samples consumed by \gst, and the output of \gstt. ``A'' and ``R'' represents \accept and \reject respectively. ``TLE'' indicates ``Time limit exceeded''.}\label{tab:all-2b}
\end{table*}

\end{document}